%% file: main.tex
\documentclass[twoside,leqno]{article}

\usepackage[letterpaper]{geometry}

\usepackage{siamproceedings}

\usepackage[T1]{fontenc}
\usepackage{amsfonts, amssymb}
\usepackage{graphicx}
\usepackage{epstopdf}
\usepackage{enumitem}
\usepackage{algpseudocode,algorithm}
\usepackage{xcolor}

\newsiamremark{remark}{Remark}
\newsiamremark{hypothesis}{Hypothesis}
\crefname{hypothesis}{Hypothesis}{Hypotheses}
\newsiamthm{claim}{Claim}

\usepackage{amsopn}

\title{On the Computation of Schrijver's Kernels}
\author{Vincent Delecroix\thanks{Univ. Bordeaux, CNRS, Bordeaux INP, LaBRI, UMR 5800, F-33400 Talence, France. Email: \email{vincent.delecroix@u-bordeaux.fr}, \email{oscar.fontaine@u-bordeaux.fr}}
  \and Oscar Fontaine\footnotemark[1]    
  \and Francis Lazarus\thanks{G-SCOP/Institut Fourier, CNRS, Université Grenoble Alpes, France. Email: \email{francis.lazarus@grenoble-inp.fr}}}

\date{}

\newcommand{\Z}{\mathbb{Z}}
\newcommand{\N}{\mathbb{N}}
\newcommand{\R}{\mathbb{R}}
\renewcommand{\le}{\leqslant}
\renewcommand{\ge}{\geqslant}
\newcommand{\multiIndex}[1]{\boldsymbol{\underline #1}}

\definecolor{definecolor}{rgb}{0,0.1,0.55}
\def\define#1{\textbf{\textcolor{definecolor}{#1}}}

\begin{document}

\maketitle

\begin{abstract}
  The geometry of a graph $G$ embedded on a closed oriented surface $S$ can be probed by counting the intersections of $G$ with closed curves on $S$. Of special interest is the map $c \mapsto \mu_G(c)$  counting the minimum number of intersections between $G$ and any curve freely homotopic to a given curve $c$.
  Schrijver [\textit{On the uniqueness of kernels}, 1992] calls $G$ a \emph{kernel} if for any proper graph minor $H$ of $G$ we have $\mu_H < \mu_G$. Hence, $G$ admits a minor $H$ which is a kernel and such that $\mu_G = \mu_H$. We show how to compute such a \emph{minor kernel} of $G$ in  $O(n^3 \log n)$ time where $n$ is the number of edges of $G$, and $g\ge 2$ is the genus of $S$. Our algorithm leverages a tight bound on the size of minimal bigons in a system of closed curves. It also relies on several subroutines of independent interest including the computation of the area enclosed by a curve and a test of simplicity for the lift of a curve in the universal covering of $S$.

  As a consequence of our minor kernel algorithm and a recent result of Dubois [\textit{Making multicurves cross minimally on surfaces}, 2024], after a preprocessing that takes $O(n^3 \log n)$ time and $O(n)$ space, we are able to compute $\mu_G(c)$  in $O(g (n + \ell) \log(n + \ell))$ time given any closed walk $c$ with $\ell$ edges. The state-of-the-art algorithm by Colin de Verdière and Erickson [\textit{Tightening non-simple paths and cycles on surfaces}, 2010] would avoid constructing a kernel but would lead to a computation of $\mu_G(c)$  in $O(g n \ell \log(n \ell))$ time (with a preprocessing that takes $O(gn\log n)$ time and $O(gn)$ space). Another consequence of the computation of minor kernels is the ability to decide in polynomial time whether two graph minors $H$ and $H'$ of $G$ satisfy $\mu_H = \mu_{H'}$. 
\end{abstract}

\section{Introduction}
\subsection{Minor Kernels of Embedded Graphs}
Let $G$ be a graph embedded on a closed surface $S$. For a closed curve $c$ in $S$, we let $\mu_G(c)$ count the minimum number of intersections between $G$ and any curve freely homotopic to $c$. Schrijver~\cite{Sch92} calls $G$ a \define{kernel} if it is minor minimal among all graphs determining the same map $c\mapsto \mu_G(c)$. In other words, contracting or deleting any edge of a kernel $G$ changes $\mu_G(c)$ for at least one curve $c$.

In a beautiful series of papers, Schrijver~\cite{Sch89,Sch91,Sch92}, and then de Graff and Schrijver~\cite{ds-chscc-95,gs-dgs-97},
studied the maps $\mu_G$ and kernels in connection with systems of curves on $S$. In particular, they established that $\mu_G(c) = \frac{1}{2} \cdot \nu_{\mathcal C}(c)$ where $\mathcal C$ is a system of closed curves depending on $G$ (but independent of $c$) and $\nu_{\mathcal C}(c)$ counts the minimum number of crossings between $\mathcal C$ and any curve freely homotopic to $c$ in general position with respect to $\mathcal C$. Here, $\mathcal C$ can be viewed as a graph (its subdivision is irrelevant), so that $\nu_{\mathcal{C}}(c)$ can be interpreted as the \define{crossing number} of the system of closed curves $\mathcal C$ and $c$. This lead Schrijver~\cite{Sch92} to make an explicit correspondence between kernels and systems of curves in minimal position in $S$. From this correspondence it follows that kernels are essentially unique: if $G$ and $H$ are two kernels on $S$ with $\mu_G=\mu_H$, then $H$ can be obtained from $G$ by a series of $\Delta Y$-exchanges (see~\cite{Sch92} for a precise definition) possibly taking the (surface) dual of the resulting graph.

The uniqueness of kernels is reminiscent of an analogous property stating that the geometry of Riemannian surfaces with negative scalar curvature is fully determined by their \define{(marked) length spectrum} proven by Otal~\cite{o-smlsc-90}. Here, the length spectrum associates to every closed curve the infimum of the Riemannian length of any curve in its free homotopy class. The function $\mu_G$ associated to an embedded graph $G$ can actually be seen as a length spectrum with respect to a discrete metric. The appropriate geometric notion that captures both the continuous and discrete settings is the definition of \emph{geodesic current} from Bonahon~\cite{b-tgotsvgc-1988}. In accordance with this terminology we will also call $\mu_G$ the \define{$\mu$-spectrum of $G$}.

For an embedded graph $G$, we say that $H$ is a \define{minor kernel} of $G$ if $H$ is a minor of $G$, a kernel and has the same $\mu$-spectrum as $G$. The effective computation of minor kernels is desirable for at least two main reasons. Thanks to the uniqueness of kernels, it allows to decide efficiently if two graphs have the same $\mu$-spectrum (see Theorem~\ref{thm:H-G-equivalence} below) and leads to an efficient computation of $\mu_G$ (see Theorem~\ref{thm:length-spectrum-computation} below). The main purpose of this paper is thus to propose an efficient algorithm that given an embedded graph $G$ computes one of its minor kernels. We assume for all our complexity analyses the standard RAM model of computation with words of logarithmic size with respect to the input. More precisely, we show the following:

\begin{theorem}\label{thm:minor-kernel}
Let $G$ be a graph with $n$ edges cellularly embedded in a closed oriented surface $S$ of genus $g\ge 2$. A minor kernel of $G$ can be computed in $O(n^3 \log n)$ time.
\end{theorem}
Remark that when $S$ is a sphere, all curves are contractible so that the only possible kernel is reduced to a vertex with no edge. Also, our approach to Theorem~\ref{thm:minor-kernel} can not be extended to the torus case. Nevertheless, using the recent work by Delecroix et al.~\cite{DELY24}, it is possible to compute a kernel (which is not necessarily a minor of $G$) with the same $\mu$-spectrum as $G$  in $O(n^2\log\log n)$ time.

We now state two applications of Theorem~\ref{thm:minor-kernel}.
\begin{theorem}
\label{thm:H-G-equivalence}
Let $G$ be a graph with $n$ edges cellularly embedded in a closed oriented surface $S$ of genus $g\ge 2$. Let $H$ be a graph minor of $G$ given as a list of minor operations on the edges of $G$. Then, we can decide whether $G$ and $H$ have the same $\mu$-spectrum in $O(n^3 \log n)$ time.
\end{theorem}
A similar but different question is to decide whether $\mu_G = \mu_H$ for two graphs $G$ and $H$ embedded on the same surface $S$. However doing so requires a way to compute an homotopy equivalence between $G$ and $H$. This situation is discussed in more details in Section~\ref{ssec:spectrum-equality-test} 
 
\begin{theorem}\label{thm:length-spectrum-computation}
Let $G$ be a graph with $n$ edges cellularly embedded in a closed oriented surface $S$ of genus $g\ge 2$. After a preprocessing in $O(n^3 \log n)$ time, we can compute $\mu_G(c)$ for every closed walk $c$ of length $\ell$ in the dual of $G$ in $O(g (n+\ell) \log(n+\ell))$ time.
\end{theorem}
It follows from previous work by Colin de Verdière and Erickson~\cite{CdVE10}
that $\mu_G(c)$ can be computed in $O\left(g n \ell \log(n \ell)\right)$ time
after a precomputation in $O(g  n \log ng)$ time and $O(gn)$ space. Our complexity improvement for the computation of $\mu_G(c)$ relies on
the precomputation of a minor kernel of $G$ in $O(n^3 \log n)$ time. We then take advantage on
Schrijver's correspondence between kernels and system of closed curves together with more recent results for computing crossing numbers of systems of closed curves by Despré and Lazarus~\cite{DL19} as improved by Dubois~\cite{Dub24}. Note, however, that~\cite{CdVE10} applies to a more general setting. See Section~\ref{ssec:relations-between-spectra} for details.

\subsection{Technical Overview and Article Organization}

There are in fact two natural ways of defining the length of a curve $c$ in terms of intersections with an embedded graph $G$ on a closed surface $S$. One may consider the minimum number of intersections between $G$ and \emph{any} curve homotopic to $c$ as for $\mu_G$. Another option is to restrict the set of admissible curves to be in \emph{general position} with respect to $G$, i.e. to be transverse to the edges of $G$ and to avoid its vertices. This second option induces another notion of length spectrum, which we denote by $\nu_G$.

As recalled in Section~\ref{sec:length-spectra}, the $\mu$-spectrum of a graph $G$ is equal to the $\nu$-spectrum (up to a constant factor 2) of its medial graph $M(G)$. The medial graph being 4-regular it can be decomposed into a system of closed curves whose crossings are the vertices of $M(G)$. The minor operations on $G$ then correspond to \emph{smoothings} of $M(G)$ viewed as a system of closed curves, which consist in deleting vertices and reconnecting their incident strands.
It appears to be more convenient to work with $\nu$-spectra and smoothings rather than with $\mu$-spectra and minors. Our proof of Theorem~\ref{thm:minor-kernel} relies on the simplification of a system of closed curves by smoothing it as much as possible without changing its $\nu$-spectrum. The resulting system of closed curves then corresponds to a minor kernel.
The main difficulty in computing a minor kernel of $G$ thus resides in the detection of smoothings of $M(G)$ that do not change the $\nu$-spectrum. To this end, we rely on the work of Schrijver~\cite{Sch91} implying that we can greedily smooth the corners of empty monogons and minimal bigons as defined in Section~\ref{sec:preliminaries}. Finding and smoothing empty monogons can easily be done in linear time. Finding minimal bigons then becomes our main goal. The proof that we can search for them efficiently lies at the core of our algorithm.

\bigskip

We first prove in Section~\ref{sec:length-bigon} that the length of minimal bigons must be linear in the number of edges of $G$ (see Theorem~\ref{thm:a-minimal-bigon-has-linear-length}). To do so, we associate with every minimal bigon a subgroup of the fundamental group of $S$. This subgroup has two generators corresponding to the sides of the bigon. As a subgroup of a surface group of genus $\ge 2$, it is necessarily isomorphic to either the trivial group, or an infinite cyclic group, or a free group of rank two. Each case requires a distinct proof. The trivial group case is straightforward. The infinite cyclic group case follows from the analysis of  (periodical) curve configurations in the plane and relies mostly on the Jordan curve theorem. The free group case is arguably the most technical one. The proof in this case proceeds by associating with each vertex on the boundary of the minimal bigon an element in the free group. Using the expression of these elements as words on the generators we argue that most of them are distinct and hence correspond to distinct vertices of $M(G)$. Since the number of vertices of $M(G)$ is the size of $G$, we finally deduce 
a linear bound on the length of minimal bigons.

\bigskip

We then develop in Section~\ref{sec:computation-area} an algorithm to compute the \emph{area} of a bigon, i.e. the number of faces circumscribed by the bigon, in time proportional to its boundary length (see Theorem~\ref{thm:linear-time-area-computation}). This area computation is a special case of an algorithm for computing the (signed) area circumscribed by the lift in the universal covering of $S$ of any contractible closed walk on $M(G)$. The algorithm relies on the standard reduction of embedded graphs to one-vertex and one-face graphs via the tree-co-tree decomposition of Eppstein~\cite{Epp02}. For these simpler graphs, we provide a discrete analogue of Stokes' theorem showing that the area we are interested in (a discrete 2-dimensional integral) is equal to a certain sum along the closed walk (a discrete 1-dimensional integral).

\bigskip

The last tool in our bigon detection algorithm is an efficient test for the simplicity of the lift of a curve of length $\ell$ in the universal covering of $S$. We  provide an $O(\ell \log(\ell+g))$ time algorithm for this task in Section~\ref{sec:simplicity-test} where $g$ is the genus of $S$. The idea beneath the algorithm is to compute geodesic representatives for \emph{all} the prefix paths of a pointed curve and observe that these geodesic representatives can be organized into a star-shaped domain.

\bigskip

The linear size bound of minimal bigons, the area computation for the interior of bigons, and the simplicity test in the universal covering are the three main ingredients of our minimal bigon detection algorithm that we develop in Section~\ref{sec:kernel-minor-algo}. The first step is to collect the balanced bigons of linear size in $M(G)$. Here, a bigon is balanced if it is bounded by two curves of the same length. To do so, we check for each vertex $v$ of $M(G)$ whether we can find a balanced bigon of linear size with corner $v$.
There is at most one such bigon per sector incident to $v$. Collecting these bigons requires $O(n \log n)$ time per sector and hence $O(n^2 \log n)$ time in total. By our linear size bound, this collection must contain a minimal bigon, if any. In order to identify one we rely on the simple observation that a balanced bigon with minimal area must be a minimal bigon. We thus compute the area of each bigon in the collection and find one with minimal area. This way, we are able to either certify that there is no bigon or find a minimal bigon in $O(n^2 \log n)$ time.

From the minimal bigon detection, it is  straightforward to compute a minor kernel as we explain in Section~\ref{sec:kernel-minor-algo}.
Namely, we pick an empty monogon or a minimal bigon, smooth one of its corners and repeat the operation until there is none left. This operation is repeated at most $n$ times (the number of vertices of $M(G)$). Hence, the total time for the computation of a minor kernel is $O(n^3 \log n)$.

\bigskip

We finally discuss in Section~\ref{sec:applications} how to use Theorem~\ref{thm:minor-kernel} to test the equality of $\mu$-spectra of a graph and one of its minors (Theorem~\ref{thm:H-G-equivalence}) and to compute $\mu_G$ (Theorem~\ref{thm:length-spectrum-computation}). 
Thanks to Schrijver's correspondence, the equality of spectra can be reduced to the problem of homotopy equivalence of systems of curves. The latter can be solved efficiently (see Theorem~\ref{thm:homotopy-system-of-curves}) using the homotopy test for curves from~\cite{LR12,EW13} (see also Section~\ref{sec:simplicity-test}) as well as standard algorithms on strings. The main cost of our algorithm for Theorem~\ref{thm:H-G-equivalence} is the kernel minor computation.
As we already mentioned, the computation of the $\mu$-spectrum is a direct application of our minor kernel computation and of the works in~\cite{DL19,Dub24} providing efficient algorithms to compute $\mu_G$ when $G$ is a kernel.

\section{Preliminaries}\label{sec:preliminaries}

We introduce in this section basic topological and combinatorial definitions that will be used throughout the text. For reference for topology of surfaces, see for example~\cite{Mas91}.

\paragraph{Combinatorial Surfaces}
Let $G=(V,E)$ be an undirected graph with vertex set $V$ and edge set $E$. We allow $G$ to have loops and multiple edges. A \define{combinatorial surface} $(S,G)$ is an embedding of $G$ on some closed oriented surface $S$ such that $S\setminus G$ is a collection of disjoint open disks. In this case, $G$ is \define{cellularly} embedded in $S$. 
A combinatorial surface can be encoded by a \define{rotation system} describing the circular orderings of the (oriented) edges around each vertex. The counterclockwise sequence of arcs (i.e., oriented edges) in the boundary of a face is called a boundary, or \define{facial walk}. See~\cite{mt-gs-01} for a general description. 
In this paper we only consider graphs embedded on oriented surfaces and often consider the embedding as implicit. 

\paragraph{Curves} A \define{curve} is a continuous map $c:[0,1]\rightarrow S$. It is \define{closed} if $c(0)=c(1)$, in which case we also write $c:\mathbb R/\mathbb Z\rightarrow S$ to emphasize this property. A curve is \emph{pointed} at $v\in S$ if $c(0)=v$. 
Two closed curves $c,c':\mathbb R/\mathbb Z\rightarrow S$ are \emph{equal as non-pointed curves} if they are equal as maps up to orientation preserving reparametrization. 
Given $c:[0,1]\rightarrow S$ and $c':[0,1]\rightarrow S$ such that $c(1)=c'(0)$, their \emph{concatenation} $c\cdot c'$ is the curve defined by $c\cdot c'(t)=c(2t)$ if $t<\frac{1}{2}$ and $c\cdot c'(t)=c'(2t-1)$ otherwise. If $c:\mathbb R/\mathbb Z\rightarrow S$ is a closed curve, we denote by $c^k$ the concatenation of $k$ copies of $c$. Moreover $c^\infty:\mathbb R\to  \mathbb R/\mathbb Z\to S$ is the composition of $c$ with the projection of $\mathbb R$ on $\mathbb R/\mathbb Z$. A curve $c:[0,1]\rightarrow S$ is \emph{simple} if $\forall i<j$ with $(i,j)\neq(0,1)$ we have $c(i)\neq c(j)$.

\paragraph{System of Closed Curves} Let $\mathcal C = (c_1,\ldots, c_k)$ be a system of closed curves on an oriented (but not necessarily closed) surface $S$, i.e. a sequence of maps $c_i:\mathbb R/\mathbb Z\rightarrow S$. The system of closed curves $\mathcal C$ is in \emph{general position} if every intersection of the curves of $\mathcal C$ is crossed transversely by exactly two curves $c_i$ and $c_j$ (possibly $i=j$). A closed curve $c$ is in \emph{general position with respect to} $\mathcal C$ if $\mathcal C\cup\{c\}$ is in general position. Similarly, two closed curves $c$ and $c'$ are in \emph{general position} if the system $(c,c')$ is in general position. An intersection point of a system of closed curves in general position is called a \define{crossing} to emphasize that the curves intersect themselves transversely.

If a system of closed curves $\mathcal{C} = (c_1, \ldots, c_k)$ on a closed surface $S$ is in general position, its image on $S$ is a $4$-regular graph $M$ whose set $V$ of vertices is the set of crossings of the curves $c_i$ in $\mathcal C$ and whose edges are the connected components of $\cup_{i=1}^k c_i\setminus V$.  Conversely, to an embedded $4$-regular graph $M$ in $S$, we can associate a system of closed curves $\mathcal C$ as follows. Each curve in $\mathcal C$ is a closed walk in $M$, each edge of $M$ is traversed exactly once by these curves and for each vertex of $M$ with incident edges $e_1,e_2,e_3,e_4$ in cyclic order, $e_1,e_3$  are traversed consecutively by a curve and so are $e_2,e_4$. Note that the system of closed curve $\mathcal{C}$ is only well defined up to a choice of ordering and orientation of these closed walks. For most considerations in this article, this choice of ordering and orientation does not matter and we often identify $4$-regular embedded graphs and systems of curves.

A system of closed curves in general position is \define{filling} if the induced graph is cellularly embedded in $S$. Note that in this case the graph must be connected. In this paper, we only consider system of closed curves in general position and filling. We may thus refer to their vertices, edges and faces as for any combinatorial surface.

\paragraph{Homotopy and Fundamental Group} Two closed curves $c$ and $c'$ on $S$ are \define{freely homotopic} if there is a continuous map $h:[0,1]\times \mathbb R/\mathbb Z\rightarrow S$ such that $h(0,t)=c(t)$ and $h(1,t)=c'(t)$.
Two curves $c: [0,1] \to S$ and $c':[0,1] \to S$ such that $c(0) = c'(0)$ and $c(1) = c'(1)$ are \define{homotopic} if there is a continuous map $h:[0,1]\times \mathbb R/\mathbb Z\rightarrow S$ such that $h(0,t)=c(t)$, $h(1,t)=c'(t)$ and both $h(\cdot,0)$ and $h(\cdot,1)$ are constant. A curve $c$ is \define{contractible} if it is freely homotopic to a constant curve. The set of homotopy classes of closed curves pointed at $v$ on $S$ endowed with the concatenation forms a group denoted  $\pi_1(S,v)$ and called the \define{fundamental group} of $S$. The set of all free homotopy classes of curves is in bijection with the conjugacy classes of $\pi_1(S,v)$ and is thus denoted by $\pi_1(S)^{conj}$.

\paragraph{Universal Covering}
A continuous map $p:S'\rightarrow S$ between two surfaces is a \define{covering map} if for every point $v\in S$ there is a small open neighborhood $v\in O$ such that $p^{-1}(O)$ is a non-empty finite or countable union of disjoint open sets $(O_i)_{i\in I}$ and the restriction of $p$ to $O_i$ is a homeomorphism for all $i\in I$. We will abusively identify the covering map $p$ with its domain $S'$.
The covering $p$ is \define{universal} if $S'$ is simply connected, i.e. every closed curve in $S'$ is contractible. Every surface has a universal covering which is unique up to homeomorphism. We will reserve the notation $p:\widetilde S\rightarrow S$ for the universal covering of $S$. An \define{automorphism} of $p: \widetilde{S} \to S$ is a homeomorphism $a: \widetilde S\to \widetilde S$ such that $p\circ a=p$.

\paragraph{Lifts of Curves}
Consider a curve $c:[0,1]\to S$ and a covering $p: S'\to S$. A \define{lift} of $c$ in $S'$ is a curve $c':[0,1]\rightarrow S'$ such that $p\circ c'=c$.
We will omit to mention the covering when it is clear from the context.
If $v$ is a point of the universal covering $p: \widetilde S\to S$ and  $c:\mathbb{R}/\mathbb Z\rightarrow S$ is a closed curve pointed at $p(v)$, then the lift of $c$ \emph{starting at $v$} is the curve $\widetilde c:[0,1]\rightarrow \widetilde S$ such that $\widetilde c(0)=v$ and $p\circ \widetilde c=c$.
%The lift of the concatenation $c^k$ starting at $v$ is denoted $\widetilde c^k$. Similarly, $\widetilde c^\infty: \mathbb R\rightarrow \widetilde S$ denotes the lift of $c^\infty$ with $\widetilde c^\infty|_{[0,1]}=\widetilde c$.
 Since the point $\widetilde c(1)$ only depends on the homotopy class of $c$ we may define a bijection between homotopy classes in $\pi_1(S,p(v))$ and automorphisms of $\widetilde S$: given a  homotopy class $\tau$ of closed curves pointed at $p(v)$ and a point $x$ in $\widetilde S$ we set $\tau(x) = (\widetilde{c\cdot p\circ r})(1)$, where $c$ is any representative of $\tau$ and $r$ is any path from $v$ to $x$ in $\widetilde S$. 
We can thus associate to a subgroup $\Gamma$ of $\pi_1(S,p(v))$ the covering
 $\widetilde{S} /\Gamma \to S \simeq \widetilde{S}/\pi_1(S,p(v))$.
 
 Now assume that $S$ has genus at least $2$. A non-contractible homotopy class of curves generates an infinite cyclic group $\Gamma=\langle \tau\rangle$. The quotient $\widetilde{S} / \langle \tau \rangle$ is then homeomorphic to the standard cylinder $Z=\mathbb R/\mathbb Z\times\mathbb R$. We denote by $p_Z: \mathbb R^2 \to Z, (x,y)\mapsto (x\mod 1, y)$ the universal covering map of $Z$.

\begin{lemma}\label{lem:simple-lift-square}
Let $c$ be a closed curve in $Z$. If $c^2$ has a  simple lift in $\R^2$, then so is every lift of $c^\infty$.
\end{lemma}
\begin{proof}
  First note that the lifts in $\R^2$ of a curve in $Z$ are (horizontal) translates of each other, so that a curve has a simple lift if and only if all its lifts are simple. Fix $v\in\R^2$ so that $p_Z(v) = c(0)$. Let $\widetilde c$ be the lift of $c$ starting at $v$, and assume by contradiction that the lift $\widetilde c^\infty$ of $c^\infty$ extending $\widetilde c$ is not simple. We can associate to the curve $c$ pointed at $p_Z(v)$ an automorphism $\theta$ of $\R^2$, in fact an integer translation along the $x$-axis, such that $\theta(v) = \widetilde{c}(1)$. Since $\theta(\widetilde c^\infty)=\widetilde c^\infty$, the lift $\widetilde c$ must contain a self-intersection of $\widetilde c^\infty$. Without loss of generality we can assume that $\widetilde c$ intersects $\theta^{j}(\widetilde c)$ for some $j> 0$,  reverting the orientation of $c$ if necessary. We let $i$ be the minimal such $j$'s. Since $c^2$ has a simple lift, we must have $i\ge 2$.
  
  Let $m$ and $M$ be points of $\widetilde c$ with minimal and maximal $y$-coordinates, respectively. We choose $m$ and $M$ such that the  subpath $g$ of $\widetilde c$ between $m$ and $M$ does not contain other points with minimal or maximal $y$-coordinate\footnote{In the trivial case where $m$ and $M$ are equal, $\widetilde{c}$ must be a horizontal segment, in which case the lemma is trivial.}. Then $\widetilde c^\infty$ is contained in the strip bounded by the horizontal lines through $m$ and $M$, and $g$ splits this strip. Since $c^2$ has a simple lift, the paths $g$ and $\theta(g)$ are disjoint, while $g$ and $\theta^i(g)$ lie on either sides of $\theta(g)$ in the strip. Using again that $c^2$ has a simple lift,  we deduce that $\widetilde{c}$ and $\theta^i(g)$ also lie on either sides of $\theta(g)$. Since $\theta^{i}(\widetilde c)$ intersects $\widetilde c$ and contains $\theta^i(g)$, we infer that $\theta^i(\widetilde c)$ must intersect $\theta(g)$. This implies that $\widetilde c$ and $\theta^{i-1}(\widetilde c)$ intersect, in contradiction with the minimality of $i$.
\begin{figure}[ht!]
    \centering
	\includegraphics[height=3cm]{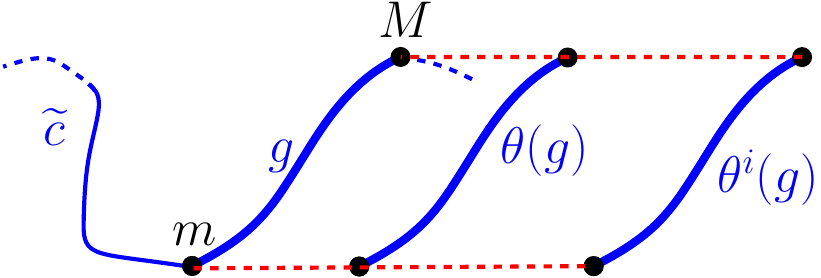}
	\caption{The region defined by $c'$}
	\label{Simple}
\end{figure}
\end{proof}
\begin{corollary}\label{cor:simple-lift-square}
  Let $c$ be a closed curve in $S$. If $c^2$ has a  simple lift in $\widetilde S$, then every lift of $c^\infty$ is simple.
\end{corollary}
\begin{proof}
As in the cylinder case, note that the lifts in $\widetilde{S}$ of a curve in $S$ are related by automorphisms of $\widetilde{S}$, so that a curve in $S$ has a simple lift in $\widetilde{S}$ if and only if all its lifts are simple. 

Let $\tau$ be the homotopy class of $c$ with basepoint $c(0)$. Remark that $p: \widetilde{S}\to S$ factors as $\widetilde{S}\xrightarrow{r} \widetilde{S}/\langle \tau\rangle \xrightarrow{q} S$, where $q, r$ are covering maps. Consider a closed lift $\overline{c}$ of $c$ in $\widetilde{S}/\langle \tau\rangle$. By the previous remark, every lift of $\overline{c}^2$ in $\widetilde S$ is a lift of $c^2$ and is thus simple. Since $\widetilde{S}/\langle \tau\rangle$ is isomorphic to the cylinder $Z$, it follows from Lemma~\ref{lem:simple-lift-square} that the lifts of $\overline{c}^\infty$ are simple. Using the same remark again, these lifts are also the lifts of  $c^\infty$, proving the lemma.
\end{proof}

\paragraph{Preimage of Graphs}
If $G$ is a graph cellularly embedded in $S$ and $p': S'\to S$ a covering of $S$, then the \emph{preimage} of $G$ is the graph $G'=p'^{-1}(G)$, which is cellularly embedded in $S'$. Note that $G'$ may have an infinite number of vertices, edges and faces. We define similarly the preimage of a system of closed curves as $\mathcal C' = p'^{-1}(\mathcal C)$. If $p'$ has infinite degree, the preimage of a closed curve in $\mathcal C$ is a union of closed curves and bi-infinite rays (i.e. curves parametrized by $\mathbb{R}$).

When the covering of $S$ is the universal covering $p: \widetilde S \to S$, the preimage of the graph $G$ is denoted by $\widetilde G=p^{-1}(G)$ and the one of the system of curves $\mathcal C$ by $\widetilde {\mathcal C}=p^{-1}(\mathcal C)$.

\paragraph{Bigons and Monogons}Let $\mathcal C=(c_1,\dots, c_k)$ be a system of closed curves in $S$ and $\widetilde {\mathcal C}$ its preimage in the universal covering of $S$. A \define{bigon} of $\widetilde{\mathcal C}$ is a pair of two interior disjoint non-trivial compact segments $U\subset \widetilde{c_i^\infty}$ and $D\subset \widetilde{c_j^\infty}$ contained in some lifts of curves $c_i$ and $c_j$ in $\mathcal{C}$ (possibly with $i=j$) so that $U \cup D$ is the boundary of a closed disk in $\widetilde S$. The segments $U$ and $D$ are called the \define{sides} of the bigon and their endpoints its \define{corners}.
A \define{monogon} of $\widetilde{\mathcal C}$ is a non-trivial compact segment of a lift of curve $U\subset \widetilde{c_i^\infty}$ with coincident endpoints bounding a disk in $\widetilde S$. The segment $U$ is called the \define{side} of the monogon and its endpoint its \define{corner}.

\paragraph{Smoothing}A \define{smoothing} at a vertex $v$ of the system of closed curves $\mathcal C$ is the system of closed curves obtained from $\mathcal C$ by removing the vertex $v$ and reconnecting the incident strands in one of two ways as in Figure~\ref{Opening}. Note that a smoothing may change the number of components.

\begin{figure}[ht!]
    \centering
	\includegraphics[height=3cm]{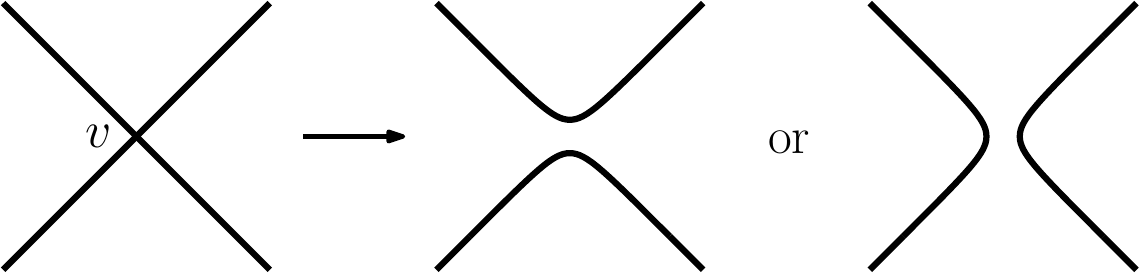}
	\caption{A smoothing at $v$}
	\label{Opening}
\end{figure}

Let $v$ be the corner of a monogon $m$ in $\widetilde{\mathcal C}$. The smoothing of $\mathcal C$ at $p(v)$ such that the two strands of $p(m)$ are not reconnected is denoted $\mathcal C_{v,m}$. We say that $\mathcal C_{v,m}$ is obtained from $\mathcal C$ by \emph{smoothing the monogon $m$}.

Similarly, given a corner $v$ of a bigon $b$ in $\widetilde{\mathcal C}$, there is a single smoothing at $v$ in $\tilde{S}$ that does not reconnect the sides of $b$ which defines by projection a smoothing at $p(v)$ in $S$. We denote the obtained system of closed curves\footnote{If the two corners $v$ and $w$ of a bigon have the same projection (i.e.: $p(v)=p(w)$), Schrijver showed in \cite{Sch91} that $C_{v,b}$ and $C_{w,b}$ coincide.} $\mathcal C_{v,b}$ and say that it is obtained from $\mathcal C$ by \emph{smoothing the bigon $b$ at $v$}.

\paragraph{Tight System of Closed Curves}
A system of closed curves $\mathcal C$ on a closed surface $S$ is \define{tight} if each curve $c_i$ in $\mathcal C$ has the minimum number of self-crossings among all closed curves freely homotopic to $c_i$ and every two $c_i$ and $c_j$, $i\neq j$ have the minimum number of crossings with each other among all closed curves freely homotopic to $c_i$ and $c_j$ respectively. A \define{tight system of primitive curves} is a tight system of closed curves where all the curves are primitive curves. The next proposition is proved in~\cite{Sch91}.

\begin{proposition}
\label{pr:Minimal2}
Let $\mathcal C$ be a system of closed curves in general position on a closed oriented surface $S$. Let $\widetilde{\mathcal C}$ be the preimage of $\mathcal C$ in the universal covering of $S$. Then $\mathcal C$ is not a tight system of primitive curves if and only one of the following holds:

\begin{enumerate}
    \item some curve in $\mathcal C$ is contractible, or
    \item $\widetilde{\mathcal C}$ has a monogon, or
    \item $\widetilde{\mathcal C}$ has a bigon.
\end{enumerate}
\end{proposition}

\paragraph{Empty Monogon and Minimal Bigon} %A component $C_i$ of a system of closed curves $\mathcal C$ is an \emph{isolated curve} if it is a contractible curve disjoint from the other curves in $\mathcal C$. 
Let $\mathcal C$ be a system of closed curves on a closed oriented surface $S$. Let $\widetilde{\mathcal C}$ be the preimage of $\mathcal C$ in the universal covering $p: \widetilde{S} \to S$. A monogon in $\widetilde{\mathcal C}$ is \define{empty} if the only vertex on its boundary is its corner. A bigon $b$ in $\widetilde{\mathcal C}$ is \define{minimal} if the disk bounded by $b$ does not contain another closed curve, a monogon or a bigon. In a minimal bigon, every curve entering through one side must exist through the opposite side. In particular, its two sides have the same length.
A more convenient version of Proposition~\ref{pr:Minimal2} is the following
\begin{proposition}\label{prop:minimal}
Let $\mathcal C$ be a filling system of closed curves in general position on a closed surface $S$ and let $\widetilde C$ be its preimage in the universal covering $\widetilde S$ of $S$. Then $\mathcal C$ is not a tight system of primitive curves if and only if one of the following holds:

\begin{enumerate}
    \item $\widetilde{\mathcal C}$ has an empty monogon, or
    
    \item $\widetilde{\mathcal C}$ has a minimal bigon.
\end{enumerate}
\end{proposition}
Note that one can formulate a slightly different proposition by omitting the filling assumption on $\mathcal C$ and adding a third case in the conclusion: the possible presence of isolated contractible curves.
\begin{proof}
The reverse implication follows from Proposition~\ref{pr:Minimal2}.
In order to prove the direct implication, we use a result from Hass and Scott~\cite{HS85}: if a closed curve on an orientable surface has self-crossings while being homotopic to a simple curve, then it includes a segment forming a monogon or two segments forming a bigon.

Suppose that $\mathcal C$ is not a tight system of primitive curves.
We first show that $\widetilde{\mathcal C}$ necessarily contains a monogon or a bigon. Assume by contradiction that $\widetilde{\mathcal C}$ does not. Then by Proposition~\ref{pr:Minimal2}, there is a curve $c_i$ in $\mathcal{C}$ that is contractible. This curve must be simple. Otherwise, by~\cite{HS85} it would include a segment forming a monogon (in $S$) or two segments forming a bigon (in $S$), in contradiction with our hypothesis. So $c_i$ is simple and bounds a topological disk. Now, if $c_i$ is intersected by another curve $c_j$, we consider a maximal arc of $c_j$  inside the disk bounded by $c_i$. On the one hand, if this arc is simple then it forms a bigon with $c_i$, in contradiction with our hypothesis.
On the other hand, if the arc is not simple it must include a segment forming a monogon or two segments forming a bigon, again leading to a contradiction.
We infer from this discussion that $c_i$ is a simple contractible curve which is not intersected by any other curve. As $S$ is not a sphere, one of its adjacent faces is not a disk. So $\mathcal{C}$ is not filling, contradicting the hypothesis.
We conclude that $\widetilde{\mathcal C}$ necessarily contains a monogon or a bigon.

We finally prove that $\widetilde{\mathcal C}$ actually contains an empty monogon or a minimal bigon.
Let us define the area of a monogon or of a bigon in $\widetilde{\mathcal C}$ as the number of faces of $\widetilde{\mathcal C}$ it encloses. We claim that a monogon or a bigon of minimal area among all monogons and bigons in $\widetilde {\mathcal C}$ is necessarily an empty monogon or a minimal bigon:
If the minimal area is realized by a bigon, then this bigon, call it $b$ can not contain any other monogon or bigon as they would have smaller area. $b$ can neither contain a closed curve since, by an argument analogous to the previous discussion, it would also form a monogon or a bigon inside $b$. It follows that $b$ is minimal.
Finally, if the minimal area is realized by a monogon $m$ we claim that it must be empty. For otherwise $m$ would be crossed by some curve $\widetilde{c_j}$ and an argument similar as above (applied in $\widetilde S$) shows the existence of a monogon or a bigon inside $m$. See Figure~\ref{fig:Monogon}.
\begin{figure}[ht!]
    \centering
    \includegraphics[height=4cm]{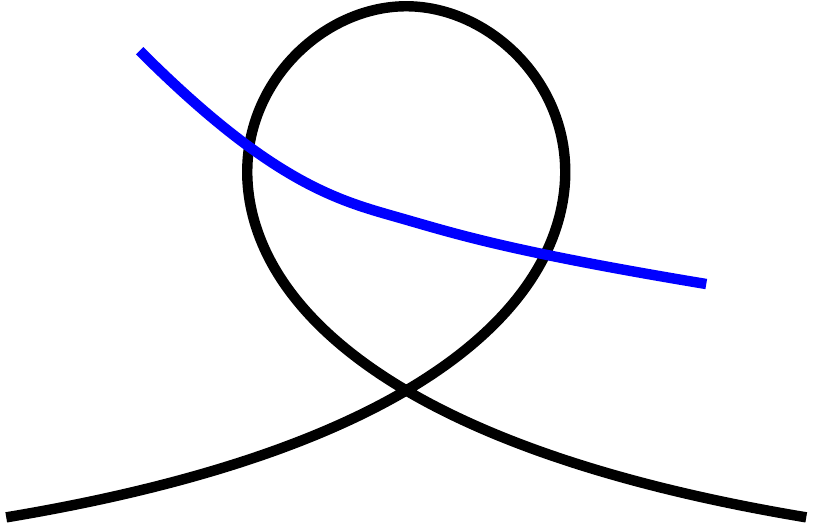}
    \caption{A non-empty monogon}
    \label{fig:Monogon}
\end{figure}
Such a monogon or bigon would have smaller area than $m$, in contradiction with the minimal area of $m$. Hence, $\widetilde {\mathcal C}$ contains an empty monogon or a minimal bigon as desired.
\end{proof}

\section{Length Spectra}\label{sec:length-spectra}
There are several natural combinatorial notions of curves in a combinatorial surface $S$ with underlying graph $G$. One may consider walks in $G$, walks in $G^*$, or sequences of diagonals in the faces of $G$. These different settings appear in the work of Schrijver~\cite{Sch91, Sch92} under a continuous disguise and give rise to different length notions for curves and length spectra of graphs.

\paragraph{The Cross Metric and $\nu$-Spectrum}
Schrijver considers continuous curves in $S$ in general position with respect to $G$. The length of such a curve $c$ is its number of crossings with $G$ and is denoted $\operatorname{cr}(c,G)$. This way of measuring length is sometimes referred to a \define{cross metric}. It is easy to see that a curve $c$ in general position with respect to $G$ can be pushed onto the edges of the graph $G^*$ dual to $G$ in $S$ so that $\operatorname{cr}(c,G)$ is the number of edges of the resulting walk in $G^*$. Conversely, every walk in $G^*$ with $\ell$ edges may be seen as a continuous curve with $\ell$ crossings with $G$. This leads Schrijver to define\footnote{$\nu_G$ is denoted $\mu'_G$ in~\cite{Sch92}.} the \define{$\nu$-spectrum} of $G$ with respect to the cross metric as the map
\[\nu_G:\pi_1(S)^{conj}\rightarrow \mathbb N, \gamma \mapsto \min_{c\in \gamma }\operatorname{cr}(c,G),
\]
where $c$ runs over all the curves representatives in the free homotopy class $\gamma$ in general position with respect to $G$.
  
\paragraph{The Vertex Metric and $\mu$-Spectrum}
Schrijver also considers continuous curves on $S$ not necessarily in general position and defines their lengths as their numbers of intersections with $G$. The \define{$\mu$-spectrum} corresponding to this \emph{vertex metric} is defined as the map
\[\mu_G:\pi_1(S)^{conj}\rightarrow \mathbb N, \gamma \mapsto \min_{c\in \gamma }\operatorname{cr}(c,G),
\]
where $c$ runs over all the curves representatives in the free homotopy class $\gamma$. This time $c$ may pass through vertices of $G$. In fact, it is easy to see that we can always choose a minimizer to intersect $G$ only at its vertices. Hence, $\mu_G(\gamma)$ coincides with the length of such a minimizer viewed as a sequence of diagonals in the faces of $G$.

\bigskip

To end this section, let us mention that in a more geometric context and at about the same time, the \define{geometric intersection number} between two free homotopy classes of curves $\gamma$ and $\gamma'$ has been introduced by Thurston in his study of the Nielsen classification of homeomorphisms of surfaces, see~\cite[Exposé 3]{FLP}. It is equal to $\min \operatorname{cr}(c,c')$ where $c$ and $c'$ runs through $\gamma$ and $\gamma'$. This geometric intersection has later been extended from curves to geodesic currents~\cite{b-tgotsvgc-1988}.

\subsection{Relations between the Two Notions of Length Spectra} \label{ssec:relations-between-spectra}
In general allowing curves in general position with respect to $G$ to pass through vertices may decrease their number of crossings with $G$, so that the $\nu$-spectrum and the $\mu$-spectrum are a priori not equal. They are however closely related.
Indeed, as already noted by Schrijver~\cite{Sch92} the $\mu$-spectrum of a graph $G$ may be obtained, up to a constant factor 2, as the $\nu$-spectrum of its medial graph\footnote{The medial graph is denoted $H(G)$ in~\cite{Sch92}.} $M(G)$: $\nu_{M(G)}=2\mu_G$. The \define{medial graph} $M(G)$ of $G$ is the graph with a vertex at the midpoint of every edge of $G$ and two midpoints are connected by an edge of $M(G)$ whenever their supporting edges in $G$ are consecutive in a face of $G$. In the other direction, we show below that any $\nu$-spectrum is equal to the $\mu$-spectrum of some other graph.

\begin{lemma}
  Given a combinatorial surface $(S,G)$, there exists a combinatorial surface $(S,H)$ such that $\mu_H=\nu_G$.
\end{lemma}
\begin{proof}
  First take the doubling $G'$ of $G$ by adding an edge parallel to every edge in $G$. Every vertex of $G'$ has even degree and $\nu_{G'}=2\nu_G$. We color in black the faces of $G'$ of length 2 corresponding to the doubling of the edges and we color in white the other faces of $G'$. We obtain this way a proper bicoloring of the faces. Decompose the graph $G'$ into curves $c_1,c_2,\dots,c_k$ so that every edge is traversed exactly once by these curves and for each 
  vertex $v$ with incident edges $e_1,e_2,\dots,e_{2d}$ in cyclic order around $v$, the edges $e_i$ and $e_{i+d}$ are traversed consecutively by a curve for all $1\le i\le d$. We next perturb the curves $c_i$ by blowing up every vertex into a non degenerate arrangement of lines inside a small neighborhood of that vertex (see Figure~\ref{blow-up}).
\begin{figure}[ht!]
    \centering
    \includegraphics[width=8cm]{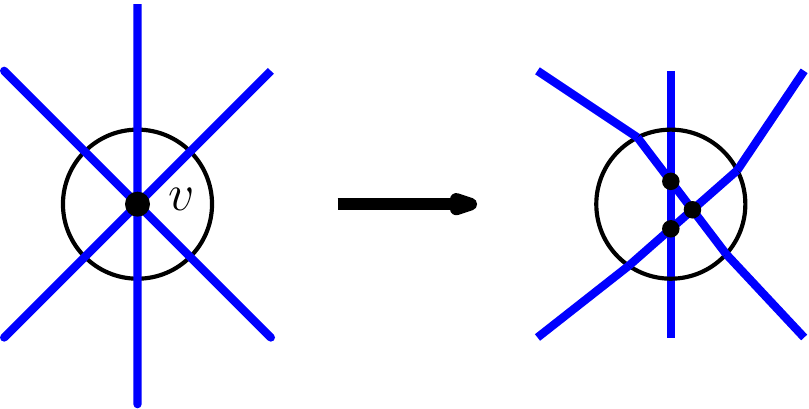}
    \caption{Blowing up a vertex.}
    \label{blow-up}
\end{figure}
The union of the perturbed curves form a graph $G''$ whose faces again admit a proper bicoloring. Moreover, we have $\nu_{G''}=\nu_{G'}$ since we can move every curve in general position with respect to $G''$ outside the vertex neighborhoods without increasing their crossing length.
We finally construct $H$ as follows. We star each white face by connecting every of its vertices to a point in the interior of that face, and finally delete all the edges of $G'$. The remaining vertices of degree two are irrelevant and may be contracted to any of their adjacent vertices. By construction $G''$ is the medial of $H$. From the above discussion we have $\nu_{G''}=2\mu_H$. It follows that $2\nu_G = \nu_{G'} = \nu_{G''}=2\mu_H$, implying the lemma.
\end{proof}
Remark that the number of vertices of the graph $H$ constructed in the proof of this lemma is $\sum_{v\in V(G)} d_v(d_v-1)/2$, where $d_v$ is the degree of $v$ in $G$. $H$ may thus have a quadratic complexity with respect to $G$. 

\subsection{Schrijver's Correspondence}
\label{sec:schrijver-s-correspondence}

Recall that a kernel $K$ on a closed oriented surface $S$ is an embedded graph $K$ such that for every proper minor $H$ of $K$, there is a closed curve $c$ satisfying $\mu_H([c])<\mu_K([c])$. If $G$ is an embedded graph, a \define{minor kernel} $K$ of $G$ is a minor of $G$ which is a kernel and such that $\mu_G=\mu_K$. Moreover, any minor operation on a graph $G$ may be obtained as a smoothing of a vertex of its medial graph $M(G)$ (See Figure~\ref{smoothing}).

\begin{figure}[ht!]
    \centering
    \includegraphics[width=13cm]{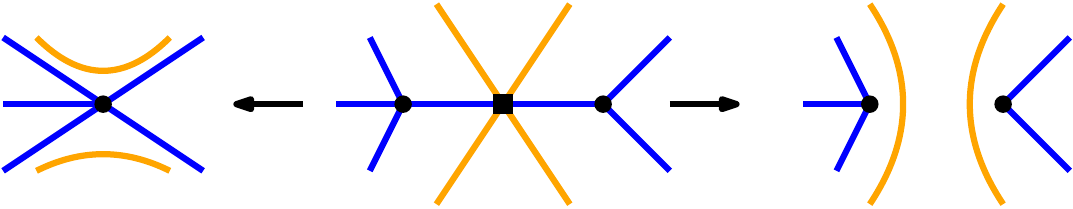}
    \caption{A graph minor operation in $G$ (blue) and the corresponding smoothings in $M(G)$ (orange)}
    \label{smoothing}
\end{figure}

Schrijver showed in~\cite{Sch92} the following correspondence between system of closed curves and graphs.

\begin{proposition}\label{pr:schrijver-s-correspondence}
Let $G$ be a graph embedded on a surface $S$. The following assertions are equivalent:
\begin{itemize}
    \item $G$ is a kernel,
    \item $M(G)$ is a tight system of primitive curves. 
\end{itemize}
\end{proposition}

We thus define the analogue notion of kernel for a system of closed curves. If $\mathcal C$ is a system of closed curves on a closed surface $S$, a \define{smoothing minimal} system $\mathcal K$ for $\mathcal C$ is a system of closed curves $\mathcal K$ obtained from $\mathcal C$ by a sequence of smoothing such that $\mathcal K$ is a tight system of primitive curves and $\nu_\mathcal{K}=\nu_\mathcal{C}$. Proposition~\ref{pr:schrijver-s-correspondence} together with the fact that any minor operation on $G$ corresponds to a smoothing in $M(G)$ allows us to reformulate the computation of a minor kernel of a graph into the computation of a smoothing minimal system of closed curves.

Recall that Proposition~\ref{prop:minimal} provides a reformulation of tightness of a system of closed curves in terms of absence of empty monogons and minimal bigons. One can actually go further and show that any such empty monogons or minimal bigons can be smoothened without changing the $\nu$-spectrum.
\begin{proposition}[\protect{Schrijver~\cite[Th. 5]{Sch91}}]
\label{Smooth spectrum}
Let $\mathcal C$ be a system of closed curves in general position on a closed surface $S$. Let $b$ be a minimal bigon of $\widetilde{\mathcal C}$ and $v$ be the projection on $S$ of one of its corners. Then $\nu_{\mathcal C}=\nu_{\mathcal C_{v,b}}$.
\end{proposition}

\begin{proposition}
Let $\mathcal C$ be a system of closed curves in general position on a closed surface $S$. Let $m$ be an empty monogon of $\widetilde{\mathcal C}$ and $v$ be the projection of its corner. Then $\nu_{\mathcal C}=\nu_{\mathcal C_{v,m}}$.
\end{proposition}

\begin{proof}
Since smoothing a vertex may only reduce the length of a curve in general position, we have $\nu_{\mathcal C}\ge \nu_{\mathcal C_{v,m}}$. We show the opposite inequality.
Let $c$ be a closed curve in $S$ in general position with respect to $\mathcal C_{v,m}$ and minimizing $\operatorname{cr}(c, \mathcal C_{v,m})$ in its homotopy class. Either $c$ avoids a small neighborhood of $v$, or $c$ crosses $v$. In the former case, $c$ is in general position with respect to $\mathcal C$ so that $\nu_{\mathcal C}([c])\le \operatorname{cr}(c, \mathcal C_{v,m})=\nu_{\mathcal C_{v,m}}([c])$.
In the latter case, each time $c$ crosses $v$, we may perturb $c$ to avoid the contractible loop $m$ without changing its crossing number (see Figure~\ref{Monogon2}). The perturbed curve $c'$ is in general position with respect to $\mathcal C$ and satisfies $\nu_{\mathcal C}([c'])\le \operatorname{cr}(c,  \mathcal C_{v,m})=\nu_{\mathcal C_{v,m}}([c])$.
\begin{figure}[ht!]
    \centering
    \includegraphics[height = 4cm]{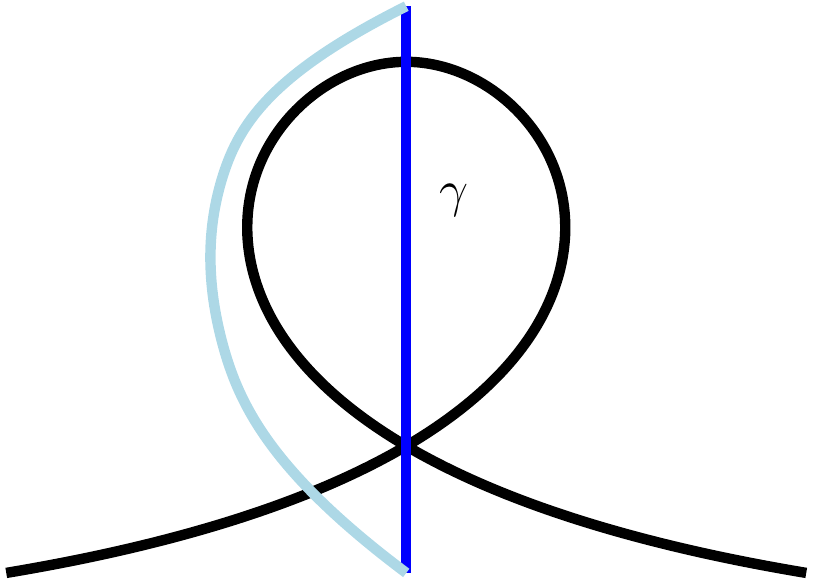}
    \caption{A curve crossing an empty monogon and its perturbation.}
    \label{Monogon2}
\end{figure}
\end{proof}
Let us remark that it may be not possible to smooth a corner of a general monogon or a general bigon without changing the length spectrum of $\mathcal C$.

\section{Length of Minimal Bigons}\label{sec:length-bigon}
Let $\mathcal C$ be a system of closed curves in general position on a closed oriented surface $S$. It induces a graph with vertex set $V$. The length of a walk $c$ in $\mathcal C$ or $\widetilde{\mathcal C}$, i.e. its number of edges, is denoted by $|c|$. We also define the length of a bigon in $\widetilde{\mathcal C}$ as the number of edges on its boundary.
The goal of this section is to show the following.
\begin{theorem}
\label{thm:a-minimal-bigon-has-linear-length}
Let $S$ be a closed orientable surface of genus $g\ge 2$. Let $\mathcal C$ be a system of closed curves on $S$. Then the length of every minimal bigon of $\widetilde{\mathcal C}$ is at most $8n$ where $n$ is the number of crossings of $\mathcal C$.
\end{theorem}
The proof of this theorem follows from Lemma~\ref{lem:a-minimal-bigon-has-linear-length} applied to Propositions~\ref{prop:abelian-group-case} and~\ref{prop:free-group-case} below. We fix some notations. In the sequel 
$b$ is a minimal bigon of $\widetilde{\mathcal{C}}$ with corners $v$ and $w$. We denote by $U$ and $D$ the two sides of $b$ and by $\widetilde{C_u^\infty}$ and $\widetilde{C_d^\infty}$ the two curves in $\widetilde{\mathcal{C}}$ supporting $U$ and $D$, respectively.
We also denote by $C_u$ and $C_d$ the two (non-periodic) curves pointed at $p(v)$ in $\mathcal{C}$ such that $C_u^\infty=p(\widetilde{C_u^\infty})$ and $C_d^\infty=p(\widetilde{C_d^\infty})$. The orientation of $D$ from $v$ to $w$ induces an orientation of $C_d$. If $C_u$ coincides with $C_d$ as non-oriented and non-pointed curves, we orient $C_u$ and $C_d$ the same way. Otherwise, we orient $C_u$ consistently with the orientation of $U$ from $v$ to $w$. These orientations induce orientations of $\widetilde{C_u^\infty}$ and $\widetilde{C_d^\infty}$. See Figure~\ref{Long} for an example. 

% Let $\widetilde C_u$ and $\widetilde C_d$ be the respective lifts of $C_u$ and $C_d$ starting at $v$. 

\begin{figure}[ht!]
    \centering
	\includegraphics[height=4cm]{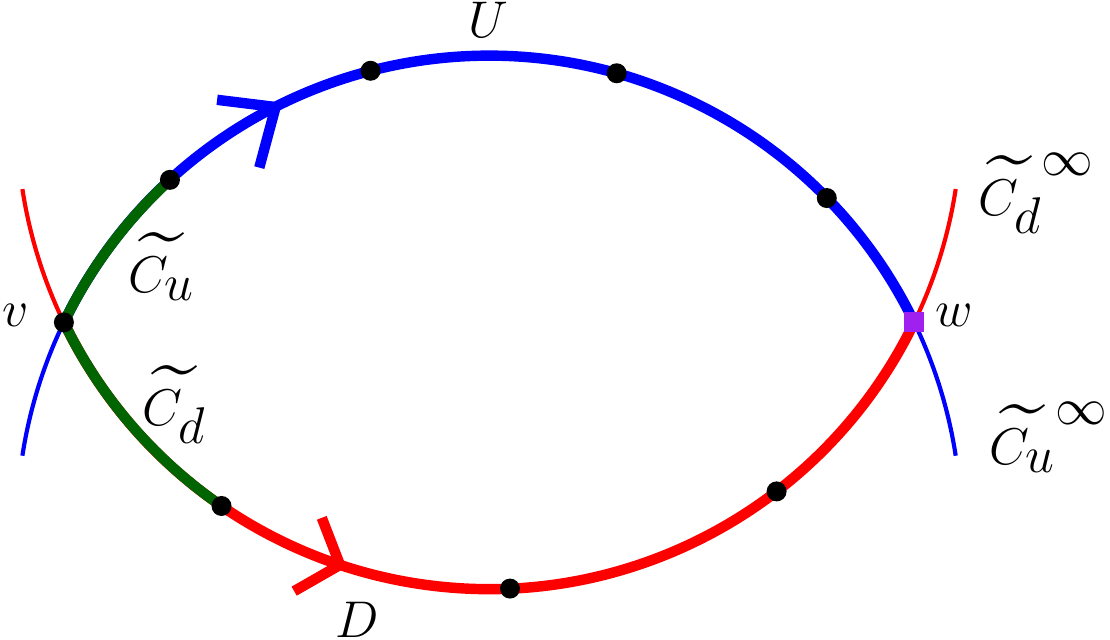}
	\hspace{0.5cm}
	\includegraphics[height=4cm]{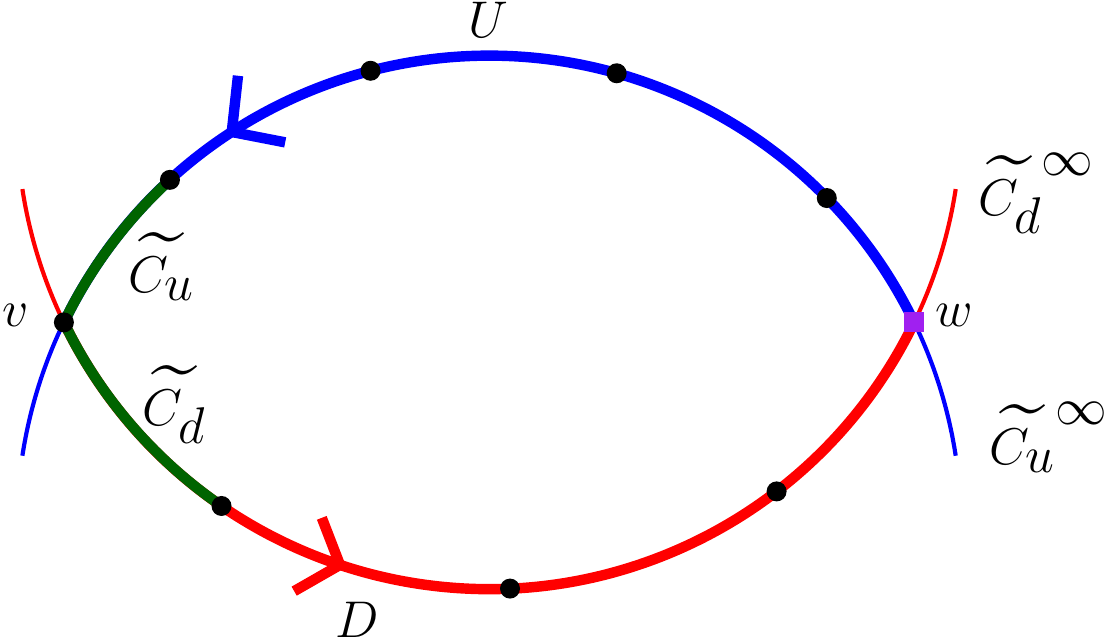}
	\caption{The respective lifts $\widetilde{C_u}$ and $\widetilde{C_d}$ of $C_u$ and $C_d$, starting at $v$. All black vertices are lifts of $p(v)$. The right figure corresponds to the case where $C_u$ and $C_d$ coincide as non-oriented and non-pointed curves.}
	\label{Long}
\end{figure}

Let $n_u=|C_u|$ and $n_d=|C_d|$ be the respective lengths of $C_u$ and $C_v$ as closed walks in $\mathcal C$.  Let $\ell_u= \lfloor (|U|-1)/n_u \rfloor$ be the number of lifts of $C_u$ that fit into $U\setminus\{w\}$. Similarly, we set $\ell_d= \lfloor (|D|-1)/n_d\rfloor$. A bound on either $\ell_u$ or $\ell_d$ bounds the length of $b$ as stated in the following lemma.

\begin{lemma}\label{lem:a-minimal-bigon-has-linear-length}
If $\ell_u\leqslant 1$ or $\ell_d\leqslant 1$ then the length of $b$ is at most $8n$.
\end{lemma}
\begin{proof}
  By hypothesis, one side of $b$ contains at most two lifts of $C_u$ or of $C_d$. Since $\mathcal C$ is in general position each curve goes at most twice through a vertex, so that $n_u\le 2n$ and $n_d\le 2n$. It follows that one side of $b$, say $U$, has at most $4n$ edges. By minimality of $b$, every curve entering through $U$ must exist through $D$. Hence $U$ and $D$ have the same length. The lemma follows.
\end{proof}

In order to prove Theorem~\ref{thm:a-minimal-bigon-has-linear-length}, it is enough to prove that $\ell_u\le 1$ or $\ell_d\le 1$. We show this in different cases. We denote by $\tau_u\in\pi_1(S,p(v))$ the homotopy class of $C_u$ based at $p(v)$. Similarly we denote by $\tau_d$ the homotopy class of $C_d$ based at $p(v)$. As $S$ is oriented and not a torus, the subgroup $\langle\tau_u,\tau_d\rangle$ of $\pi_1(S,p(v))$ generated by these two elements is isomorphic either to the free group $F_2$ of rank $2$, or to $\mathbb Z$, or to the trivial group. We consider these three cases separately.

\subsection{The trivial group case}
If $\langle\tau_u,\tau_d\rangle$ is the trivial group, then $C_u$ and $C_d$ are contractible in $S$. Equivalently $\widetilde{C_u}$ and $\widetilde{C_d}$ are closed curves in $\widetilde S$. Since $U$ is a simple path, it must be a subset of $\widetilde{C_u}$. It follows that $\ell_u = 0$.
In fact, since each vertex of $V$ appears at most twice on $U$, and since the same is true for $D$, we deduce that the length of $b$ is at most $4n$.

\subsection{The $\Z$ case}
If $\langle \tau_u,\tau_d \rangle \simeq \mathbb Z$, there is a primitive element $\tau\in\pi_1(S,p(v))$ such that $\tau_u=\tau^{k_u}$ and $\tau_d=\tau^{k_d}$ for some integers $k_u$ and $k_d$. 

\begin{proposition}
\label{prop:abelian-group-case}
If $\langle \tau_u,\tau_d \rangle \simeq \mathbb Z$, then $\ell_u\le 1$ or $\ell_d\le 1$.
\end{proposition}
\begin{proof}
  The proof proceeds by case analysis on $k_u$ and $k_v$.
 \begin{itemize}
      \item If $k_u=0$ or $k_d=0$, then $\tau_u$ or $\tau_d$ is contractible. In this case, $U$ or $D$ contains less than a lift of $C_u$ or $C_d$, whence $\ell_u=0$ or $\ell_d=0$.  This proves the proposition in this case. 
      \item If $k_u = k_d\neq 0$, then $\tau_u=\tau_d$. It follows that $\widetilde{C_u^\infty}$ and $\widetilde{C_d^\infty}$ cross at $\tau_u\cdot v=\tau_d\cdot v$, implying $\ell_u\le 1$ and $\ell_d\le 1$. This proves the proposition also in that case.
      \item If $k_u = - k_d\neq 0$, then $\tau_u=\tau_d^{-1}$, implying that $\tau_u(\widetilde{C_d^\infty}) = \widetilde{C_d^\infty}$ crosses $\tau_u(\widetilde{C_u^\infty}) = \widetilde{C_u^\infty}$ at $\tau_u(v)$. Assume for a contradiction $\ell_u>1$ and $\ell_d>1$. Then, the lift of $C_d^2$ from $v$ being part of $D$ it must be simple and it follows from Corollary~\ref{cor:simple-lift-square} that $\widetilde{C_d^\infty}$ is also simple. It also ensues from $\ell_u>1$ that $\tau_u(v)$ lies in the relative interior of $U$. Since every crossing is transverse, this means that $\widetilde{C_d^\infty}$ enters (or exit) the bigon $b$ at $\tau_u(v)$. By the Jordan curve theorem $\widetilde{C_d^\infty}$ must cross the boundary of $b$ once again. Being simple $\widetilde{C_d^\infty}$ cannot cross $D$, so it crosses $U$ again. This however creates a smaller bigon included in $b$ in contradiction with the minimality of $b$.
      \item It remains to consider the case where $k_u$ and $k_v$ are distinct nonzero integers. Exchanging the roles of $C_u$ and $C_d$, or reverting their orientations if necessary, we may assume $0<|k_d|< k_u$. We again argue by contradiction, assuming that $\ell_u>1$ and $\ell_d>1$. As above this implies that $\widetilde{C_u^\infty}$ and $\widetilde{C_d^\infty}$ are simple curves. We consider the translation $\theta: \R^2\to\R^2, (x,y)\mapsto (x+1,y)$. Viewing $\tau$, pointed at $p(v)$, as an automorphism of $\widetilde S$, we let $\varphi: \widetilde S\to \R^2$ be an (equivariant) homeomorphism satisfying\footnote{To construct $\varphi$, one can endow $S$ with a hyperbolic metric. Then $\tau$ acts as a hyperbolic translation of the hyperbolic plane $\widetilde S$. We can then set $\varphi(\tau^y(v_x)) := (x,y)$, where $v_x$ is the point with abscissa $x$ on the hyperbolic line through $v$ orthogonal to the axis of $\tau$, and $\tau^y$ is the hyperbolic translation with the same axis as $\tau$ and (algebraic) translation length $y$.} $\varphi\circ\tau=\theta\circ \varphi$. It appears more convenient to straighten the lifts  $\widetilde{C_u^\infty}$ and $\widetilde{C_d^\infty}$ by taking their images by $\varphi$. We still denote these images by  $\widetilde{C_u^\infty}$ and $\widetilde{C_d^\infty}$ and remark that they are left globally invariant by $\theta^{k_u}$ and $\theta^{k_d}$, respectively. Also remark that $\theta^{k_d}(\widetilde{C_u^\infty})$ and $\widetilde{C_u^\infty}$ cannot coincide. Otherwise, $\widetilde{C_u^\infty}$  would be $|k_d|$-periodical. However, as $|k_d|<k_u$, this would imply that two edges of $\widetilde{C_u}$ project onto a same edge, in contradiction with the fact that $C_u$ is traversed only once by definition. In turn, this implies that $\theta^{k_d}(\widetilde{C_u^\infty})$ and $\widetilde{C_u^\infty}$ may only have transverse crossings. The same is true for $\theta^{-k_d}(\widetilde{C_u^\infty})$ and $\widetilde{C_u^\infty}$.

    Let $m$ and $M$ be points of $\widetilde{C_u}$ with minimal and maximal $y$-coordinates, respectively. We choose $m$ and $M$ such that the  subpath $g$ of $\widetilde{C_u}$ between $m$ and $M$ does not contain other points with minimal or maximal $y$-coordinate. Then $\widetilde{C_u^\infty}$ is contained in the strip $B$ bounded by the horizontal lines through $m$ and $M$. Remark that $g$ is a simple path that divides $B$ into two connected components.  Changing the orientation of the $y$-axis if necessary, we may assume that $m$ comes first along $\widetilde{C_u}$. We denote by $H_u$ the connected components of $\mathbb R^2\setminus \widetilde{C_u^\infty}$ containing the bigon $b$. We also denote by $\widetilde{C_u^2}$ the lift of $C_u^2$ starting at $v$.
    \begin{claim}\label{cl:2-crossings}
      There is a subpath of either $\theta^{k_d}(\widetilde{C_u^2})$ or $\theta^{-k_d}(\widetilde{C_u^2})$ whose relative interior is contained in $H_u$, and whose endpoints lie in the relative interior of $\widetilde{C_u^2}$.
    \end{claim}
    \begin{proof}[Proof of the claim]
      First observe that $k_u$ and $k_d$ being distinct, the curves $C_u$ and $C_d$ are also distinct. It follows that the lifts of $C_u$ can either coincide or cross away from the lifts of $p(v)$. In particular, $\theta^{k_d}(\widetilde{C_u^2})$ and $\widetilde{C_u^2}$ can only cross in their relative interiors. The same is true for $\theta^{-k_d}(\widetilde{C_u^2})$ and $\widetilde{C_u^2}$. There are  two cases to consider.
      \begin{itemize}
      \item If $H_u$ lies below $\widetilde{C_u^\infty}$, then $H_u$ lies to the right of $g$. Denote by $h$ the subpath of $\widetilde{C_u^2}$ between $M$ and $\theta^{k_u}(m)$.  On the one hand, since $\theta^{-|k_d|}(m)$ and $\theta^{k_u-|k_d|}(m)$ lie respectively to the left and right of $g$, the relative interior of the path $\theta^{-|k_d|}(g\cup h)$ must cross $g$ from left to right, thus entering $H_u$. On the other hand, $\theta^{k_u-|k_d|}(m)$ lies to the left of $\theta^{k_u}(m)$, while $\theta^{k_u-|k_d|}(M)$ lies to the right of $M$. It follows that $\theta^{k_u-|k_d|}(g)$ crosses $h$ (possibly at $\theta^{k_u-|k_d|}(m)$ or $\theta^{k_u-|k_d|}(M)$). We infer from these two crossings the existence of a subpath of $\theta^{-|k_d|}(g\cup h\cup \theta^{k_u}(g))$ with relative interior contained in $H_u$ and with endpoints on $g\cup h$. Since $g\cup h\cup \theta^{k_u}(g)\subset \widetilde{C_u^2}$, this proves the claim in this case.
      \item If $H_u$ lies above $\widetilde{C_u^\infty}$, then $H_u$ lies to the left of $\theta^{k_u}(g)$.
        On the one hand, since $\theta^{|k_d|}(M)$ and $\theta^{k_u+|k_d|}(m)$ lie respectively to the left and right of $\theta^{k_u}(g)$, we infer that the relative interior of $\theta^{|k_d|}(h)$ must cross $\theta^{k_u}(g)$ from left to right, thus exiting $H_u$. On the other hand, since $M$ and $\theta^{k_u}(m)$ lie respectively to the left and right of $\theta^{|k_d|}(g)$, we have that $h$ and $\theta^{|k_d|}(g)$ intersect. We infer from these two crossings the existence of a subpath of $\theta^{|k_d|}(g\cup h)$ with relative interior contained in $H_u$ and with endpoints on $h\cup \theta^{k_u}(g)$. This proves the claim also in this case, and ends the proof of the claim.
      \end{itemize}
    \end{proof}

\begin{figure}[ht!]
    \centering
	\includegraphics[height=4cm]{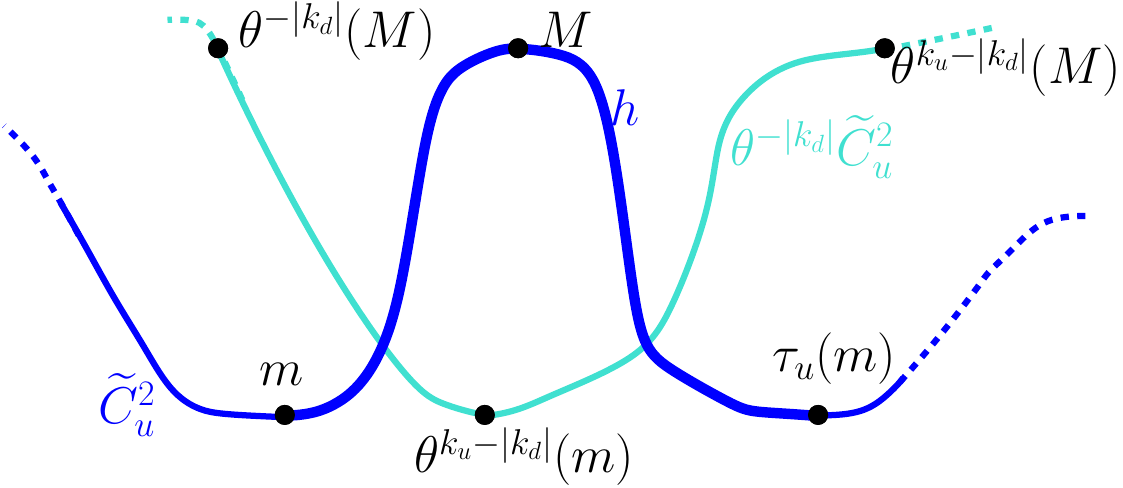}
	\caption{The bigon in $H_u$ when $H_u$ is below $\widetilde{C^\infty_u}$.}
	\label{HBigon}
\end{figure}
Now, remark that $\ell_u\ge 2$ implies that the path $\widetilde{C_u^2}$ is a part of $U$, so that its relative interior avoids $\widetilde{C_d^\infty}$. Since $\widetilde{C_d^\infty}$ is globally invariant under the action of $\theta^{k_d}$, the relative interiors of $\theta^{k_d}(\widetilde{C_u^2})$ and $\theta^{-k_d}(\widetilde{C_u^2})$ also avoid $\widetilde{C_d^\infty}$.

By Claim~\ref{cl:2-crossings} there is a subpath $\pi$ of either $\theta^{k_d}(\widetilde{C_u^2})$ or $\theta^{-k_d}(\widetilde{C_u^2})$ with endpoints in the relative interior of $U$, and otherwise contained in $H_u$. By the previous remark $\pi$ does not cross $\widetilde{C_d}^\infty$, hence does not cross $D$. It follows from the Jordan curve theorem that $\pi$ is contained in $b$, thus forming a bigon with the subpath of $U$ between the endpoints of $\pi$. This however contradicts the minimality of $b$. Hence the assumption $\ell_u>1$ and $\ell_d >1$ cannot occur, ending the proof of Proposition~\ref{prop:abelian-group-case}.
\end{itemize}
\end{proof}

\subsection{The $F_2$ Case}
We now assume $\langle \tau_u,\tau_d\rangle\simeq F_2$ and fix an orientation of $S$ once for all.
\begin{proposition}
\label{prop:free-group-case}
If $\langle \tau_u,\tau_d \rangle \simeq F_2$, then $\ell_u\le 1$ or $\ell_d\le 1$.
\end{proposition}

\begin{proof}
  We index the vertices of $C_u$ by $\Z/n_u\Z$ so that
  $C_u(0),C_u(1),\ldots,C_u(n_u-1)$ is the sequence of vertices along $C_u$. Note that each vertex that is a self-crossing of $C_u$ appears twice in this list. We similarly index the vertices of $C_d$ by $\Z/n_d\Z$. In the case $C_u=C_d$ as non-oriented and non-pointed curves, we use the same indexation, i.e. $C_u(i)=C_d(i)$ for all $i$. The indexation of $C_u$ determines an indexation of any lift of $C_u^\infty$ up to a shift by a multiple of $n_u$. The same holds for the indexations of  $C_d$ and of the lifts of $C_d^\infty$.
 If $s$ is the index of a vertex of $\widetilde{C_u^\infty}$ we denote by $\overline{s}=s+n_u\Z$ the class of $s$ in $\Z/n_u\Z$, which we also view as an element of $\{0,1,\dots n_u-1\}$. We similarly write $\overline{s}=s+n_d\Z$ if $s$ is an index along $\widetilde{C_d^\infty}$.
  Let $s_0$ and $t_0$ be the indices of  $v$ on $\widetilde{C_u^\infty}$ and $\widetilde{C_d^\infty}$, respectively. 
  
  To every index $s\in \Z/n_u\Z$ we associate the path $g_s := C_u([\overline{s_0},s])$. Similarly for every $t\in \Z/n_d\Z$ we set $h_t := C_d([\overline{t_0},t])$.  We then define $\gamma_{s,t}^U\in \pi_1(S, p(v))$ for every $(s,t)\in \Z/n_u\Z\times \Z/n_u\Z$ such that $C_u(s)=C_u(t)$ as the homotopy class of the loop $g_s\cdot g_t^{-1}$. We also define $\gamma_{s,t}^D$  for every $(s,t)\in \Z/n_u\Z\times \Z/n_d\Z$ such that $C_u(s)=C_d(t)$ as the homotopy class of the loop $g_s\cdot h_t^{-1}$. 

\begin{figure}[ht!]
    \centering
	\includegraphics[height=5cm]{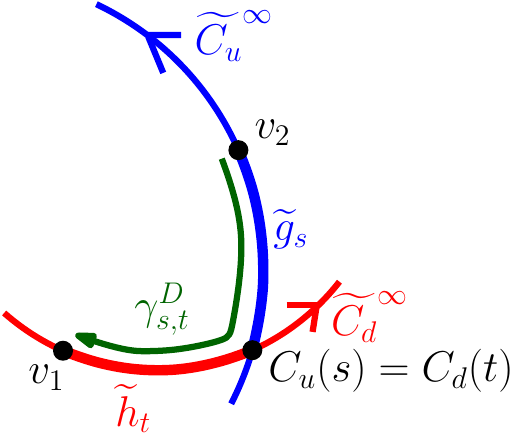} \hspace{1 cm} \includegraphics[height=5cm]{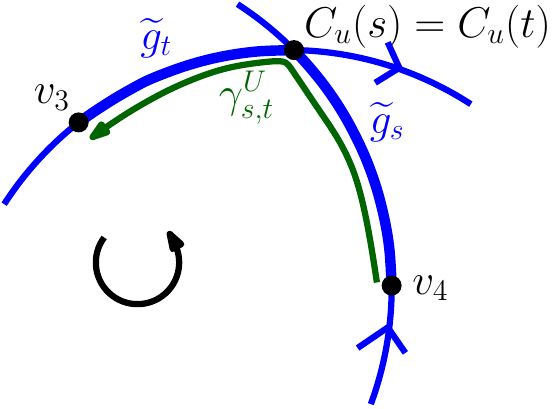}
	\caption{Definition of $\gamma^U_{s,t}$ and $\gamma^D_{s,t}$  in the universal covering.}
	\label{Gamma}
\end{figure}

Let $v_i$ be the vertex with index $t_0+i.n_u$ on $\widetilde{C_d^\infty}$. Hence, $v_0=v$ and 
$v_1,\ldots, v_{\ell_d}$ lie in the relative interior of $D$. Fix $1\le i\le\ell_d$. There is a unique lift $\widetilde{C_i}$ of $C_u^\infty$ crossing $D$ at $v_i$. By minimality of $b$, this lift must cross $U$ at some vertex $y_i$. Let $s_i, t_i$ be the indices of $y_i$ along $\widetilde{C_i}$ and $U\subset\widetilde{C_u^\infty}$, respectively. The lift of $g_{\overline{s_i}}\cdot g_{\overline{t_i}}^{-1}$ passing through $y_i$ may be completed in a closed curve formed by the segment of $D$ between $v$ and $v_i$ then the segment of $\widetilde{C_i}$ from $v_i$ to $y_i$ and the segment of $U$ between $y_i$ and $v$ (see Figure~\ref{Longmax}). It ensues that $\gamma_{\overline{s_i},\overline{t_i}}^U=\tau_u^{\alpha_i}\tau_d^{-i}\tau_u^{\beta_i}$ for some $\alpha_i, \beta_i\in\Z$.
\begin{figure}[ht!]
    \centering
	\includegraphics[height=4cm]{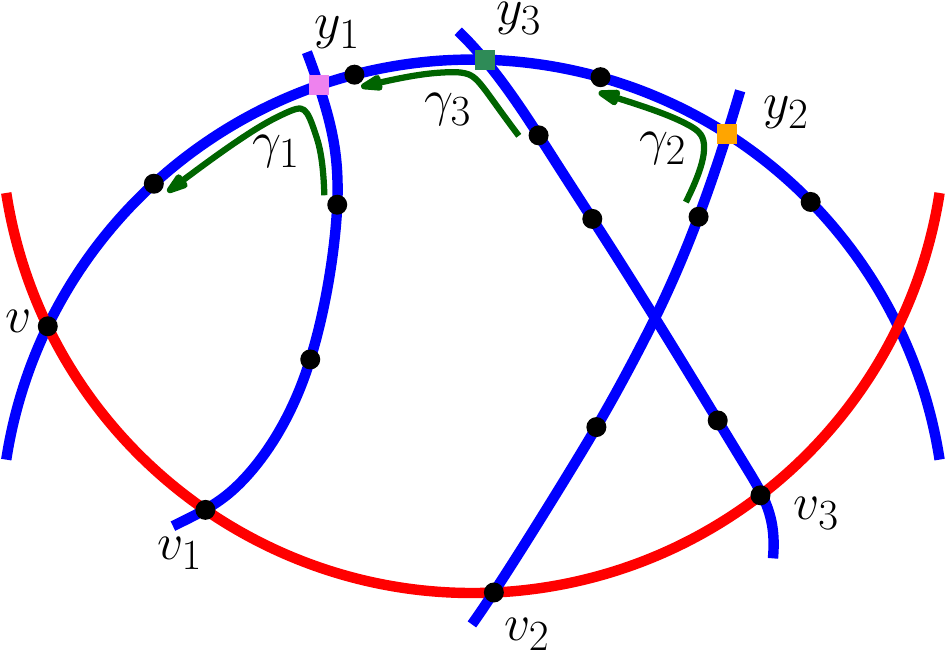}
	\caption{The $y_i$'s and $\gamma_i$'s. All dark dots are lifts of $v$.}
	\label{Longmax}
\end{figure}
From the uniqueness of reduced forms in a free group, we deduce $\gamma_{\overline{s_i},\overline{t_i}}^U\neq \gamma_{\overline{s_j},\overline{t_j}}^U$. Since each $(\overline{s_i},\overline{t_i})$ corresponds to a self-crossing of $C_u$, the index $\overline{t_i}$ is determined by $\overline{s_i}$ and we infer that $\ell_d\le n_d \le 2 n$. To prove the proposition, we shall iterate the same sort of argument.

For a multi-index $\multiIndex{i}=(i_0,\ldots,i_m)\in \N^{m+1}$ and an index $j$, we set $\multiIndex{i}\cdot j:=(i_0,\ldots,i_m,j)$. Moreover, we let $|\multiIndex{i}|:=m+1$ denote the length of $\multiIndex{i}$,  and for $q\in\{0,\ldots m\}$ we define $\multiIndex{i}(q)$ as the restriction of $\multiIndex{i}$ to its first $q+1$ values.
We recursively define $y_{\multiIndex{i}}$ as an intersection of some lift of $C_u^{\infty}$ with either $U$ (when $|\multiIndex{i}|$ is even) or $D$ (when $|\multiIndex{i}|$ is odd) by the following construction. These definitions will be valid for all $\multiIndex{i}$ such that $1\le i_0\le \ell_d$, and such that $i_q$ lies in some interval $J_{\multiIndex{i}}$ of size $\ell_d$ or $\ell_d+1$ when $q$ is even, and of size $\ell_u$ or $\ell_u+1$ when $q$ is odd.
%$0\le i_q< \ell_d$ if $q$ is even and not $0$ and $0\le i_q< \ell_u$ if $q$ is odd.
\begin{itemize}
\item If $|\multiIndex{i}|=1$, then $y_{\multiIndex{i}}=y_{i_0}$ as defined previously.
  \end{itemize} 
Suppose that $y_{\multiIndex{i}}$ is already defined. We define $y_{\multiIndex{i}\cdot j}$ for $j$ such that $\multiIndex{i}\cdot j$ is a valid index.
  \begin{itemize}
  \item If $|\multiIndex{i}|$ is odd, then by assumption $y_{\multiIndex{i}}$ is a crossing of $U$ with some lift of $C_u^\infty$. Let $t_{\multiIndex{i}}$ be the index of $y_{\multiIndex{i}}$ along $U\subset \widetilde{C_u^\infty}$. We let $z_{j}=\widetilde{C_u^\infty}(t_{\multiIndex{i}}+j.n_u)$ be the $j$'th ``copy'' of $y_{\multiIndex{i}}$ along $\widetilde{C_u^\infty}$. We denote by $J_{\multiIndex{i}}$ the set of $j$'s such that $z_j$ lies in the relative interior of $U$. For $j\in J_{\multiIndex{i}}$ the vertex $z_j$ is crossed by some lift of $C_u^\infty$ that must cross $D$ by minimality of $b$. We finally define $y_{\multiIndex{i}\cdot j}$ as this last crossing (see left Figure~\ref{Ys}).

\begin{figure}[ht!]
    \centering
	\includegraphics[height=4.3cm]{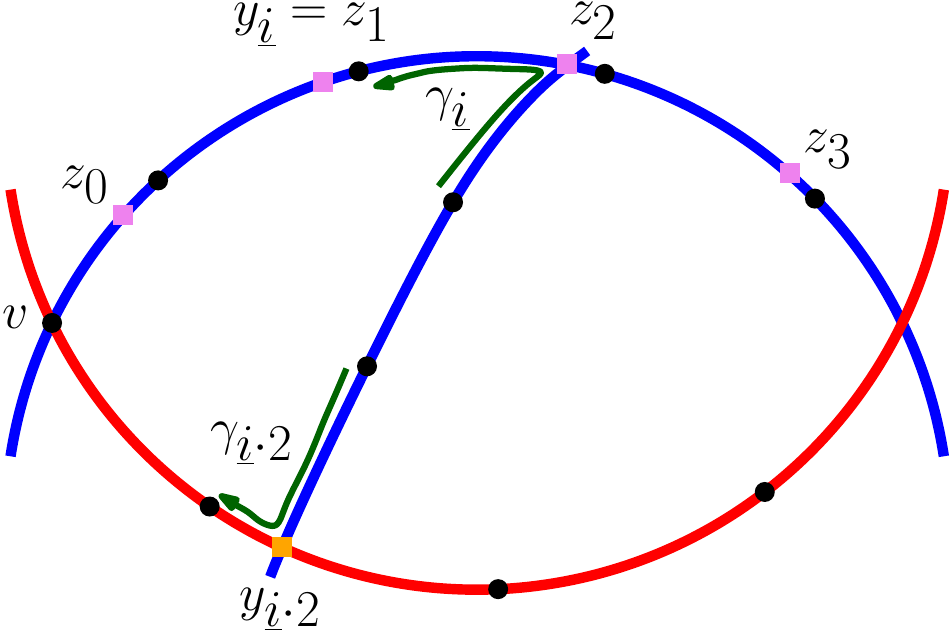}\hspace{0,5cm} \includegraphics[height=4.3cm]{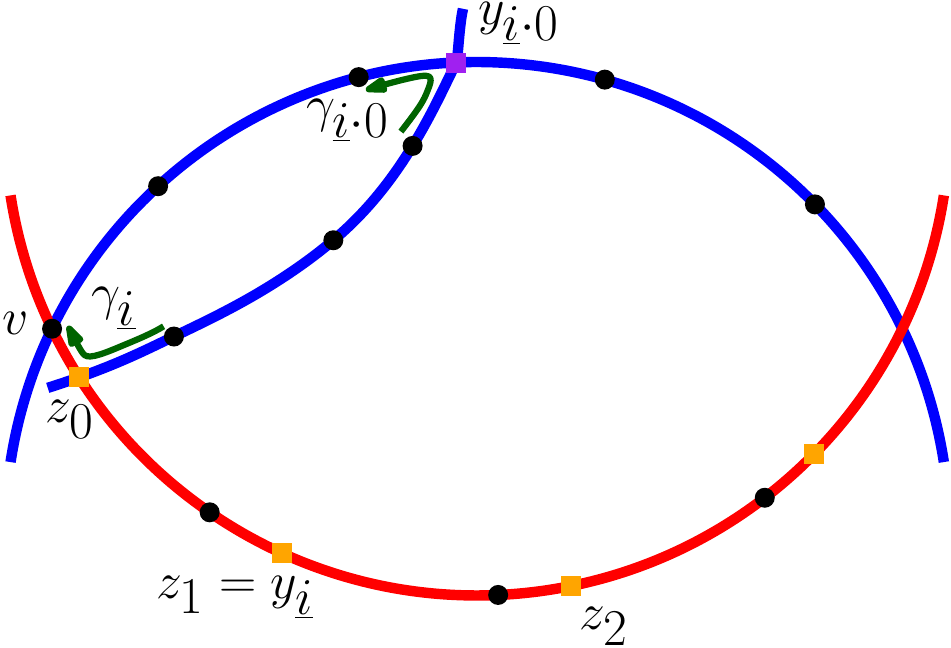}
	\caption{Definition of $y_{\multiIndex{i}\cdot j}$ from $y_{\multiIndex{i}}$}
	\label{Ys}
\end{figure}

\item Similarly, if $|\multiIndex{i}|$ is even, then $y_{\multiIndex{i}}$ is a crossing of $D$ with some lift of $C_u^\infty$. Let $s_{\multiIndex{i}}$ be the index of $y_{\multiIndex{i}}$ along $D\subset \widetilde{C_d^\infty}$. We let $z_{j}=\widetilde{C_d^\infty}(s_{\multiIndex{i}}+j.n_d)$ be the $j$'th ``copy'' of $y_{\multiIndex{i}}$ along $\widetilde{C_d^\infty}$. We again denote by $J_{\multiIndex{i}}$ the set of $j$'s such that $z_j$ lies in the relative interior of $U$. For such indices $z_j$ is crossed by some lift of $C_u^\infty$ that also crosses $U$ by minimality of $b$. We let $y_{\multiIndex{i}\cdot j}$ be this crossing (see left Figure~\ref{Ys}).
\end{itemize}
We now define $\gamma_{\multiIndex{i}}\in \pi_1(S, p(v))$ for all valid multi-indices.
\begin{itemize}
\item If $|\multiIndex{i}|$ is odd, then $y_{\multiIndex{i}}$ lies at the intersection of $\widetilde{C_u^\infty}$ with some other lift of $C_u^\infty$. We let $s_{\multiIndex{i}}, t_{\multiIndex{i}}$ be the respective indices of $y_{\multiIndex{i}}$ along these two lifts. 
We then put $\gamma_{\multiIndex{i}} = \gamma^U_{\overline{s_{\multiIndex{i}}}, \overline{t_{\multiIndex{i}}}}$.
  \item If $|\multiIndex{i}|$ is even, then $y_{\multiIndex{i}}$ lies at the intersection of some lift of $C_u^\infty$ with $\widetilde{C_d^\infty}$. We let $s_{\multiIndex{i}}, t_{\multiIndex{i}}$ be the respective indices of $y_{\multiIndex{i}}$ along these two curves. 
We then put $\gamma_{\multiIndex{i}} = \gamma^D_{\overline{s_{\multiIndex{i}}}, \overline{t_{\multiIndex{i}}}}$.
\end{itemize}
The definition of $\gamma_{\multiIndex{i}}$ induces the following recursive formulas, where $\alpha_{\multiIndex{i}\cdot j}$, $\beta_{\multiIndex{i}\cdot j}$ and $f_{\multiIndex{i}}(j)$ are integers:
\begin{itemize}
    \item If $|\multiIndex{i}|$ is odd, $\gamma_{\multiIndex{i}\cdot j}=\tau_u^{\alpha_{\multiIndex{i}\cdot j}}\gamma_{\multiIndex{i}}\tau_u^{f_{\multiIndex{i}}(j)}\tau_d^{\beta_{\multiIndex{i}\cdot j}}$. Here, $f_{\multiIndex{i}}(j)$ has the form $j_{\multiIndex{i}}-j$ or $j_{\multiIndex{i}}+j$ (for some constant $j_{\multiIndex{i}}$) depending on whether the orientation of $\widetilde{C_u^\infty}$ agrees with the orientation of $U$ from $v$ to $w$, or not.
    
    \item If $|\multiIndex{i}|$ is even, $\gamma_{\multiIndex{i}\cdot j}=\tau_u^{\alpha_{\multiIndex{i}\cdot j}}\gamma_{\multiIndex{i}}\tau_d^{f_{\multiIndex{i}}(j)}\tau_u^{\beta_{\multiIndex{i}\cdot j}}$. Here, $f_{\multiIndex{i}}(j)$ has the form $j_{\multiIndex{i}}-j$.
\end{itemize}
So if $|\multiIndex{i}|$ is odd, $\gamma_{\multiIndex{i}\cdot j\cdot h}=\tau_u^{\alpha_{\multiIndex{i}\cdot j\cdot h}+\alpha_{\multiIndex{i}\cdot j}}\gamma_{\multiIndex{i}}\tau_u^{-j}\tau_d^{\beta_{\multiIndex{i}\cdot j}-h}\tau_u^{\beta_{\multiIndex{i}\cdot j\cdot h}}$.
By induction, we can write an explicit formula for $\gamma_{\multiIndex{i}}$ with $|i|=2m+1$:
\begin{align}
    \label{Equa} \gamma_{\multiIndex{i}}=\tau_u^{\sum_{q=0}^{2m}\alpha_{\multiIndex{i}(q)}}\tau_d^{-i_0}\left(\prod_{q=1}^m\tau_u^{\beta_{\multiIndex{i}(2q-2)}-f_{\multiIndex{i}(2q-2)}(i_{2q-1})}\tau_d^{\beta_{\multiIndex{i}(2q-1)}-f_{\multiIndex{i}(2q-1)}(i_{2q})}   \right)\tau_u^{\beta_{\multiIndex{i}}}
\end{align}
Let $I_m$ denote the set of all valid $\multiIndex{i}$ of length $2m+1$, i.e. such that 
\begin{itemize}
  \item $1\le i_0\le \ell_d$,
\item $\forall q\in\{1,\ldots, 2m\}$, $i_q\in J_{\multiIndex{i}(q-1)}$, and
  \item $\forall q\in\{1,\ldots, 2m\}$, $f_{\multiIndex{i}(q-1)}(i_q)\neq\beta_{\multiIndex{i}(q-1)}$,
  % \item $0\le i_q< \ell_d$ if $q$ is even and not $0$, and
  % \item  $0\le i_q< \ell_u$ if $q$ is odd.
\end{itemize}
Then the following claim holds :
\begin{claim}\label{cl:Gamma_m-injective}
The map $\Gamma_m:I_{m}\rightarrow \pi_1(S,x),\, \multiIndex{i}\mapsto\gamma_{\multiIndex{i}}$ is injective.
\end{claim}
\begin{proof}[Proof of the claim]
  Let $\multiIndex{i}$ and $\boldsymbol{\underline j}$ be two elements of $I_{m}$ such that $\gamma_{\multiIndex{i}}=\gamma_{\boldsymbol{\underline j}}$. By Formula \eqref{Equa} and by definition of $I_m$, the expression of $\gamma_{\multiIndex{i}}$ and $\gamma_{\boldsymbol{\underline j}}$ in the basis $\tau_u,\tau_d$ are reduced. By uniqueness of reduced form in a free group, we must have $i_0=j_0$ and for all $1\le q\le 2m$, \[\beta_{\multiIndex{i}(q-1)}-f_{\multiIndex{i}(q-1)}(i_{q}) = \beta_{\multiIndex{j}(q-1)}-f_{\multiIndex{j}(q-1)}(j_{q}).
    \]
Let us prove by induction that $\multiIndex{i}(q)=\boldsymbol{\underline j}(q)$ for all $q\le 2m$.
We already know that $i_0=j_0$, i.e. $\multiIndex{i}(0)=\boldsymbol{\underline j}(0)$. 
For the inductive step, let $0< q\le 2m$. Assume that $\multiIndex{i}(q-1)=\boldsymbol{\underline j}(q-1)$. Then $\beta_{\multiIndex{i}(q-1)}=\beta_{\boldsymbol{\underline j}(q-1)}$. So $f_{\multiIndex{i}(q-1)}(i_{q}) =f_{\multiIndex{j}(q-1)}(j_{q})$, implying $i_q=j_q$. Whence $\multiIndex{i}(q)=\boldsymbol{\underline j}(q)$ as desired.
It follows that $\multiIndex{i}=\multiIndex{i}(2m)=\boldsymbol{\underline j}(2m)=\boldsymbol{\underline j}$, which proves the claim.
\end{proof}

\begin{claim}%Ou we claim that si la preuve fait une ligne.
\label{cl:Im lower bound}
For all $m\ge 0$, we have $2n \ge |I_{m}|\ge \ell_d(\ell_u-1)^m(\ell_d-1)^m$.
\end{claim}
\begin{proof}[Proof of the claim]
  Since $\Gamma_m(\multiIndex{i}) = \gamma^U_{\overline{s_{\multiIndex{i}}}, \overline{t_{\multiIndex{i}}}}$ and since every index pair $(\overline{s_{\multiIndex{i}}}, \overline{t_{\multiIndex{i}}})$ corresponds to a self-crossing of $C_u(\overline{s_{\multiIndex{i}}}) = C_u(\overline{t_{\multiIndex{i}}})$, the map $\Gamma_m$ may take at most $2n$ distinct values. By Claim~\ref{cl:Gamma_m-injective}, the map $\Gamma_m$ is injective. It follows that $|I_{m}|\le 2n$, thus proving
  the left inequality in the claim.
  
  We use induction on $m$ for the right inequality. We have $|I_{0}|=|\{1\le i_0\le \ell_d\}|=\ell_d$, implying the claim for $m=0$. For the inductive step, suppose that $|I_{m}|\ge \ell_d(\ell_u-1)^m(\ell_d-1)^m$ for some $m\ge 0$. Let $\multiIndex{i}\in I_m$. Then for any $i_{2m+1}\in J_{\multiIndex{i}}$ such that $i_{2m+1}\neq \beta_{\multiIndex{i}}$ and for any $i_{2m+2}\in J_{\multiIndex{i}\cdot i_{2m+1}}$ such that $i_{2m+2}\neq \beta_{\multiIndex{i}\cdot i_{2m+1}}$, we have $\multiIndex{i}\cdot i_{2m+1}\cdot i_{2m+2}\in I_{m+1}$. As previously noted,  $J_{\multiIndex{i}}$ contains $\ell_d$ or $\ell_d+1$ indices and, similarly, $|J_{\multiIndex{i}\cdot i_{2m+1}}|\ge \ell_u$. It ensues that  $|I_{m+1}|\ge |I_m|(\ell_u-1)(\ell_d-1)\ge \ell_d(\ell_u-1)^{m+1}(\ell_d-1)^{m+1}$, concluding the induction step.
\end{proof}

By Claim~\ref{cl:Im lower bound}, we must have $\ell_d(\ell_u-1)^m(\ell_d-1)^m\le 2n$ for all $m\ge 0$. The only possible values for $(\ell_u,\ell_d)$ are thus $(2,2)$, $(*,1)$, $(1,*)$, $(0,*)$ and $(*,0)$, where the star denotes an arbitrary natural number. Hence, Proposition~\ref{prop:free-group-case} will be proved if we can exclude the case $(\ell_u,\ell_d)=(2,2)$.
To this end, first note that for $\multiIndex{i}\in I_m$ the expression in \eqref{Equa} is reduced, ignoring the leftmost or rightmost factors whenever $\sum_{q=0}^{2m}\alpha_{\multiIndex{i}(q)}$ or $\beta_{\multiIndex{i}}$ is zero. Depending on these last two conditions, this reduced expression may have $2m+1$, $2m+2$ or $2m+3$ factors. So, if $k\ge m+2$, then\footnote{In fact, we also have $\Gamma_m(I_m)\cap\Gamma_{m+1}(I_{m+1})=\varnothing$.} $\Gamma_m(I_m)\cap\Gamma_{k}(I_k)=\varnothing$. It ensues that the map $\Gamma: \bigcup_{k=0}^\infty I_{2k}\to \pi_1(S,x),\multiIndex{i}\mapsto\gamma_{\multiIndex{i}}$ is injective. Now, if $\ell_u=\ell_d=2$, then $|I_m|\ge 2$ for all $m\in\mathbb N$, and $\bigcup_{k=0}^\infty I_{2k}$ contains infinitely many multi-indices.  However, this contradicts the fact that the injective map $\Gamma$ may take at most $2n$ distinct values. We conclude that the case $(\ell_u,\ell_d)=(2,2)$ does not occur, thus ending the proof of Proposition~\ref{prop:free-group-case}.
\end{proof}

\section{Computation of the Area of a Bigon}
\label{sec:computation-area}
We describe an efficient algorithm to compute the area of a bigon of $\widetilde{\mathcal C}$ for a system of closed curves $\mathcal C$ in $S$ in general position. More generally, our formula applies to any closed walk $c$ in $\widetilde{\mathcal{C}}$ and computes some signed area of $c$ described below. Up to the sign, this signed area coincides with the area enclosed by $c$ when $c$ is simple.

Let $H$ be a locally finite connected graph cellularly embedded in the plane. By locally finite, we mean that each vertex and face of $H$ have a finite degree. For every face $f$ of $H$ and every  (oriented) closed walk $c$ of $H$, the winding number $wind(c,f)$ of $c$ around $f$ is the number of times $c$ turns counterclockwise around $f$~\cite{Wiki:WN}. Formally, $wind(c,f)$ is the degree of the application $c:S^1\rightarrow \mathbb R^2\backslash\{x\}$ where $x$ is any interior point of $f$. Note that $wind(c,f)=0$ for all faces $f$ of $H$ except for a finite number of them. We may thus define the \define{signed area} of $c$ by $A(c)=\sum_{f}wind(c,f)$ where the sum is over all faces of $H$.
Now, consider a contractible closed walk in the graph $G$ of a combinatorial surface.
Any of its lifts $\widetilde c$ is a closed walk in $\widetilde G=p^{-1}(G)$ (recall that $p$ is the universal covering map) and we define the \define{signed area} of $c$ as the signed area of $\widetilde c$. Since the lifts of $c$ are related by automorphisms of $\widetilde G$, this definition does not depend on the choice of the lift $\widetilde c$.

\begin{theorem}\label{thm:linear-time-area-computation}
Let $(S,G)$ be a combinatorial surface of genus $g\ge 2$. After a precomputation in $O(n)$ time, the signed area of any oriented contractible closed walk $c$ of length $\ell$ in $G$ can be computed in $O(\ell)$ time.
\end{theorem}

We describe an algorithm to compute the signed area inspired by Erickson and Fox and Lkhamsuren~\cite{EFL18}. The proof of Theorem~\ref{thm:linear-time-area-computation} follows from Proposition~\ref{pr:correct2} stating that the algorithm is correct and from its complexity analysis established in Proposition~\ref{prop:area}. We first discuss the case where $G$ is a reduced graph with a single vertex and a single face. In this case, the signed area of a closed walk $c$ can be computed from $c$ itself without traversing the regions it encompasses. For this we rely on a discrete Stokes formula\footnote{We actually refer to the following version of Stokes formula. Let $R$ be a planar region and let $\omega$ and $\eta$ be two $1$-forms on $R$ such that $\omega=d\alpha$ for some function $\alpha$, and such that $d\eta=0$. Then, $\int_R\omega\wedge\eta=\int_{\partial R}\alpha\eta$. When $\omega\wedge\eta$ is the volume form, the right-hand size computes the area by contour integral. Our notations in~\ref{ssse:reduced} intend to suggest a parallel with the discrete case.}. We then reduce the case of general graphs to the reduced case using a tree-cotree decomposition~\cite{Epp02}.

\subsection{The Reduced Graph Case}
\label{ssse:reduced}

We first discuss the computation of the signed area in the case where $G$ has only one vertex and only one face. Let $v$ be this vertex and $L$ be the set of edges of $G$. By Euler formula, $L$ consists of $2g$ simple loop-edges of $S$ such that $S\backslash L$ is an open disk. Then $p^{-1}(L)$ is the $1$-skeleton of a tiling of $\widetilde S$ by fundamental domains. We denote by $\overset{\leftrightarrow}{L}$ the set of arcs of $L$ ; it contains two opposite arcs per edge of $L$. Consider a domain $D$  in this tiling and let $(e_1,\ldots, e_{4g})$ be the list of arcs of the counterclockwise boundary $\partial D$ of $D$, starting from an arbitrary vertex. 

A \define{co-boundary} $\omega:\overset{\leftrightarrow}{L}\rightarrow \mathbb R$ is a map such that $\omega(\overline e)=-\omega(e)$ where $\overline e$ is the arc opposite to $e$. 
For concision, we identify the arcs of $\partial D$ with their projections, thus writing $\omega(e_i)$ for $\omega(p(e_i))$.
Since in a reduced graph every arc and its opposite appears on $\partial D$, we have $\sum_{e\in\partial D}\omega(e)=0$, justifying the terminology.
We define the \define{wedge product} of two co-boundary functions $\omega$ and $\eta$ as the number:
\[\omega\wedge\eta=
  \sum_{1\le i<j\le 4g} \big(\omega(e_i)\eta(e_j)-\omega(e_j)\eta(e_i)\big).
\]

\begin{lemma}
\label{le:wedge}
The definition of $\omega \wedge \eta$ is independent of the choice of the starting vertex used to list the edges of $\partial D$, namely for any $k \in \mathbb Z$ we have :
\begin{align}
\label{eq:wedge}
\omega\wedge\eta=\sum_{1\le i<j\le 4g}\big(\omega(e_{i+k})\eta(e_{j+k})-\omega(e_{j+k})\eta(e_{i+k})\big) 
\end{align}
where indices are considered modulo $4g$.
\end{lemma}

\begin{proof}
We prove \eqref{eq:wedge} by induction on $k$. The case $k=0$ is a tautology. Suppose \eqref{eq:wedge} holds for some $k\ge 0$. To simplify the notation, we put $\omega_i=\omega(e_i)$ and $\eta_i=\eta(e_i)$ for all $1\le i\le 4g$, and we introduce the quantities $\zeta_{i,j} = \omega_i\eta_j - \omega_j\eta_i$.
Let $P_k=\{(i+k, j+k)\mid 1\le i<j\le 4g\}$, where the indices in the pairs are taken modulo $4g$. In other words, $P_k = \{(i+k,j+k)\mid (i,j)\in P_0\}$. With these notations, the right-hand side of \eqref{eq:wedge} writes $\sum_{(i,j)\in P_k} \zeta_{i,j}$. We claim that 
\[P_k = \big(P_{k+1}\setminus \{(1+k,j+k)\mid 2\le j\le 4g\}\big)\cup \{(j+k,1+k)\mid 2\le j\le 4g\}
\]
To see this, let $(i+k,j+k)\in P_k$ with $(i,j)\in P_0$. If $i\ge 2$, then $(i-1,j-1)\in P_0$, whence $(i+k,j+k)=((i-1)+(k+1),(j-1)+(k+1))\in P_{k+1}$. Moreover, $(1,j)\in P_0\iff (j-1,4g)\in P_0$, so that $(1+k,j+k)\in P_k\iff (j+k,1+k) = ((j-1)+(k+1),4g + (k+1))\in P_{k+1}$. The claim follows. We then compute
\begin{align*}
    \sum_{(i,j)\in P_k} \zeta_{i,j} = \sum_{(i,j)\in P_{k+1}} \zeta_{i,j} - \sum_{j=2}^{4g} \zeta_{1+k,j+k} + \sum_{j=2}^{4g} \zeta_{j+k,1+k}
    = \sum_{(i,j)\in P_{k+1}} \zeta_{i,j} - 2 \sum_{j=2}^{4g} \zeta_{1+k,j+k}
\end{align*}
where we used $\zeta_{j,i} = -\zeta_{i,j}$ for the second equality. Since, $\zeta_{i,i}=0$, we also compute
\begin{align*} 
\sum_{j=2}^{4g} \zeta_{1+k,j+k} &= \sum_{j=1}^{4g} \zeta_{1+k,j+k} = \sum_{j=1}^{4g} (\omega_{1+k}\eta_{j+k} - \omega_{j+k}\eta_{1+k})\\
&= \omega_{1+k} (\sum_{j=1}^{4g} \eta_{j+k}) - (\sum_{j=1}^{4g} \omega_{j+k})\eta_{1+k} = 0
\end{align*}
where we used that $\omega$ and $\eta$ are co-boundaries for the last equality. 
We infer that $\sum_{(i,j)\in P_k} \zeta_{i,j} = \sum_{(i,j)\in P_{k+1}}\zeta_{i,j}$, showing that \eqref{eq:wedge} holds at order $k+1$. The lemma follows.
\end{proof}

Let $\widetilde{V}=p^{-1}(\{v\})$ be the set of vertices of $\widetilde G$. A co-boundary $\omega$ derives from a \define{potential} function $\alpha:\widetilde V\rightarrow \mathbb R$ if for every arc $v_1v_2$ in the tiling: $\omega(p(v_1v_2))=\alpha(v_2)-\alpha(v_1)$.

\begin{lemma}
\label{unique}
Every co-boundary derives from a potential which is unique up to a constant function.
\end{lemma}

\begin{proof}
Let $\omega$ be a co-boundary. Let us fix $v_0\in\widetilde V$. We set $\alpha(v_0)=0$ and for every $x\in\widetilde{V}$ we let $\alpha(x)=\sum_{e\in P}\omega(p(e))$ for some oriented path $P$ from $v_0$ to $x$. Note that $\alpha(x)$ does not depend on the choice of $P$ since $\omega$ is a co-boundary (it sums up to zero on every closed walk). By construction, $\omega$ derives from the potential $\alpha$.

For the uniqueness, let $\omega$ derives from a potential $\alpha$ and from another potential $\alpha'$. Let $v_0,v_1\in \widetilde V$. By an easy induction, $\alpha(v_1)=\alpha(v_0)+\sum_{e\in P}\omega(p(e))$ for any oriented path $P$ from $v_0$ to $v_1$ (such a path exists due to the connectivity of $\widetilde G$). The same identity is true for $\alpha'$. Thus $\alpha(v_1)-\alpha'(v_1)=\alpha(v_0)-\alpha'(v_0)$. showing that $\alpha$ and $\alpha'$ differ by a constant function.
\end{proof}

We now choose two co-boundaries $\omega$ and $\eta$ satisfying $\omega\wedge\eta=2$. Since $D$ projects to the unique face of $S$, there must be four arcs $e_i, e_j, e_k, e_\ell$ in its boundary $\partial D=(e_1,e_2,\dots,e_{4g})$, such that $e_k=\overline{e_i}$ and $e_\ell=\overline{e_j}$ with $1\le i<j<k<\ell\le 4g$ (see~\cite[Sec 1.3.5]{Sti93}).
We define $\omega$ and $\eta$ by
\[\omega(p(e_i)) = - \omega(p(e_k))=1 \quad\text{and}\quad \omega(p(e_t))=0, \, \forall t\neq i,k\]
\[\eta(p(e_j))=-\eta(p(e_\ell))=1 \quad\text{and}\quad \eta(p(e_t))=0, \, \forall t\neq j,\ell\]
An easy computation shows that $\omega\wedge\eta=2$.

\begin{proposition}[Discrete Stokes formula]
\label{alpha}
Let $\omega$ derives from a potential $\alpha$. Let $c$ be a closed walk in $\widetilde G$. We have:

\begin{align}
    \label{eq:Stokes}
    A(c)=\sum_{e=xy\in c}\frac{\alpha(x)+\alpha(y)}{2}\eta(p(e)) 
\end{align}

\end{proposition}

\begin{proof}
First observe that this identity is true for the trivial walk $c=(v_0)$. Indeed, $A(c)=0$, and the sum on the right-hand side being over an empty set is also equal to zero. 
In order to prove the proposition, we rely on the fact that any closed walk in the plane can be obtained from the trivial walk by a sequence of elementary moves. There are two types of moves to consider (and their inverses):
\begin{itemize}
    \item the spur insertion, corresponding to inserting in a closed walk an arc $a$ followed by $\overline{a}$,
    \item the face boundary insertion, corresponding to inserting in a closed walk the sequence of arcs of the boundary of a face in clockwise or counterclockwise order.
\end{itemize}
We prove that these moves preserves the identity~\eqref{eq:Stokes}.
Let $c$ be a closed walk in $\widetilde G$.
\begin{itemize}
    \item 
Suppose that $c'$ is obtained from $c$ by the insertion of the spur $a\overline a$ for some arc $a$. Then, $wind(c,f)=wind(c',f)$ for any face $f$ in $\widetilde G$, whence $A(c)=A(c')$. Moreover, since $\eta(p(a))=-\eta(p(\overline a))$, the sum over $c$ or $c'$ in the right-hand side of~\eqref{eq:Stokes} is the same. Hence, $c$ satisfies \eqref{eq:Stokes} if and only if $c'$ satisfies it.
\item
Now suppose that $c'$ is obtained from $c$ by the insertion of the boundary of some face $f$, say in counterclockwise order. Then, the winding number being additive with respect to the concatenation of curves, we have $wind(c',f)=wind(c,f)+1$. Moreover, for any face $f'\neq f$ of $\widetilde G$, $wind(c',f')=wind(c,f')$. It follows that $A(c')=A(c)+1$. Let $(v_0,\varepsilon_1,v_1,\ldots,v_{4g-1},\varepsilon_{4g},v_{4g})$ be the boundary walk of $f$, oriented counterclockwise. Then, writing $\eta(e)$ instead of $\eta(p(e))$ for concision, and similarly for $\omega$, we have 
\begin{align}\label{eq:face-boundary-case}
    \sum_{e=xy\in c'}\frac{\alpha(x)+\alpha(y)}{2}\eta(e) - \sum_{e=xy\in c}\frac{\alpha(x)+\alpha(y)}{2}\eta(e) =\sum_{i=1}^{4g}\frac{\alpha(v_{i-1})+\alpha(v_i)}{2}\eta(\varepsilon_i)
\end{align}
As $\omega$ derives from $\alpha$, we have for all $1\le i\le 4g$, \[\alpha(v_i)=\alpha(v_0)+\sum_{j=1}^{i}\omega(\varepsilon_j)=\alpha(v_0)-\sum_{j=i}^{4g}\omega(\varepsilon_j).
\]
Note that the right-hand side of formula~\eqref{eq:Stokes} is left unchanged when adding a constant function to $\alpha$. We can thus assume $\alpha(v_0)=0$, and compute \begin{align*}
\sum_{i=1}^{4g}\frac{\alpha(v_{i-1})+\alpha(v_i)}{2}\eta(\varepsilon_i) &=\frac{1}{2}\sum_{i=1}^{4g}\Big(\sum_{j=1}^{i}\omega(\varepsilon_j)-\sum_{j=i+1}^{4g}\omega(\varepsilon_j)\Big)\eta(\varepsilon_i) \\
%&=\frac{1}{2} \Big( \sum_{i=1}^{4g}\sum_{j=1}^{i}\omega(\varepsilon_j)\eta(\varepsilon_i) - \sum_{i=1}^{4g}\sum_{j=i+1}^{4g}\omega(\varepsilon_j)\eta(\varepsilon_i)\Big)\\
&= \frac{1}{2} \Big( \sum_{1\le j<i\le 4g} \omega(\varepsilon_j)\eta(\varepsilon_i) -  \sum_{1\le i<j\le 4g}\omega(\varepsilon_j)\eta(\varepsilon_i)\Big)\\
&= \frac{1}{2} \sum_{1\le i<j\le 4g}\Big( \omega(\varepsilon_i)\eta(\varepsilon_j) -  \omega(\varepsilon_j)\eta(\varepsilon_i)\Big)
\end{align*}
Now, every face of $\widetilde G$ projecting to the same face as $D$, we have by Lemma~\ref{le:wedge} that this last sum does not depend on the starting arc used to insert $f$. It follows that the right-hand side of \eqref{eq:face-boundary-case} is equal to $\frac{1}{2}\omega\wedge\eta$, which is 1 by our choice of $\omega $ and $\eta$. 
It ensues that $c$ satisfies identity~\eqref{eq:Stokes} if and only if $c'$ satisfies it. A similar computation and conclusion holds in the case of the insertion of a clockwise boundary.
\end{itemize}
The proposition now follows by induction on the number of elementary moves to transform the trivial path to $c$.
\end{proof}

\subsection{The non-Reduced Graph Case}

We now consider the case where $G$ is a graph cellularly embedded in $S$ with no other restriction. Denote by $E$ the set of edges of $G$. Let $T$ be a spanning tree of $G$, let $C^*$ be a spanning tree of the dual $G^*$ of $G$ satisfying $T\cap C=\varnothing$, and let $L=E\backslash(C\cup T)$ be the set of leftover edges. The triple $(T,L,C)$ is called a \define{tree-cotree decomposition}~\cite{Epp02} and can be computed in linear time. By Euler's formula, $|L|=2g$. Moreover, $S$ cut through the edges of $T\cup L$ is a topological disk.

Let $G'=G/T$ be the graph obtained from $G$ by contracting all the edges of $T$. Let $A(\cdot)$ and $A'(\cdot)$ denotes the signed areas in respectively $\widetilde G$ and $\widetilde{G'}$. As already noted, $A({\widetilde c})$ does not depend on the chosen lift ${\widetilde c}$ of a contractible closed walk $c$. We thus define $A(c)$ as  the signed area $A({\widetilde c})$. Similarly, we set $A'(c') = A'(\widetilde{c'})$ for a contractible closed walk $c'$ in $G'$ with lift $\widetilde{c'}$ in $\widetilde{G'}$. We denote by $c/T$ the contraction in $G'$ of a walk $c$ in $G$. We claim that for every contractible closed walk $c$ in $G$, we have
\begin{align}\label{eq:contract-T}
  A(c)=A'(c/T). 
\end{align}
Indeed, there is a deformation retraction of the plane that maps $\widetilde G$ onto $\widetilde G'$. This retraction sends faces of $\widetilde G$ to faces of $\widetilde G'$ preserving the winding number of the faces. The claim follows by summation of the winding numbers over all the faces and over their retracted images.

Note that $G'$ is a graph with a single vertex $v$ and that the subgraph $G_L$ of $G'$ induced by the edges of $L$ is a reduced graph. We may thus apply  the results of Section~\ref{ssse:reduced} to compute $\omega$ and $\eta$ for $G_L$ in $O(g)$ time. Similarly to the above, we denote by $A_L(\cdot)$ the signed area of a closed walk in $\widetilde G_L$, and set $A_L(c)=A_L(\widetilde{c})$ for a contractible closed walk in $G_L$. Now, for each arc $e$ of $G'$ we define the following:
\begin{itemize}
    \item $c_e$ is the path in $L$ obtained as the subpath of the (unique) facial walk of $G_L$ with the same endpoints as $e$, and lying to the right of $e$ (see Figure~\ref{domain2}). When $e$ is a loop arc, $c_e$ can be either the trivial walk or the whole facial walk, depending on its orientation. When $e\in \overset{\leftrightarrow}L$, we simply have $c_e=e$.
    \item $A(e)$ is the positive area of the region bounded by $e$ and $c_e$. In particular, $A(e)=0$ for every $e\in \overset{\leftrightarrow}L$.
    \item $\omega(e):=\sum_{e'\in c_e}\omega(e')$
    \item $\eta(e):=\sum_{e'\in c_e}\eta(e')$
    \item $I(e):=\sum_{xy\in \widetilde c_e}\frac{\alpha_\omega(x)+\alpha_\omega(y)}{2}\eta(p(xy))$, where $\widetilde c_e$ is a lift of $c_e$ with initial vertex $v_0$ and $\omega$ derives from the potential $\alpha_\omega$ such that $\alpha_\omega(v_0)=0$. Note that for $e\in \overset{\leftrightarrow}L$ this summation reduces to $I(e)=\frac{\omega(e)}{2}\eta(e)$.
\end{itemize}
\begin{figure}[ht!]
    \centering
    \includegraphics[height=4cm]{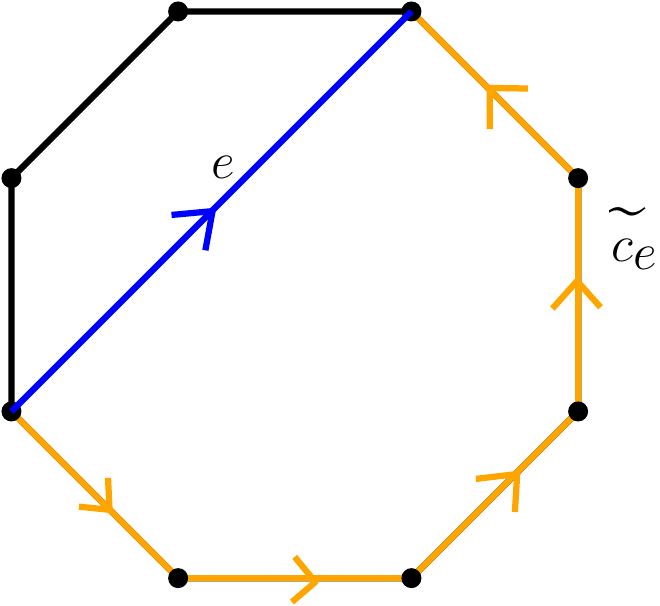}
    \caption{The path in $L$ corresponding to $e$}
    \label{domain2}
\end{figure}
After pre-computing  $A(e), \omega(e), \eta(e), I(e)$ for each arc $e$ of $G'$, Algorithm~\ref{Algo area} below describes how to compute the signed area $A(c)$ of any contractible closed walk $c$ in $G$.
\begin{algorithm}[ht!]
\caption{Area computation}
\label{Algo area}
\begin{algorithmic}[1]
  \Require A combinatorial surface with graph $G$ and $|F|$ faces together with a tree-cotree decomposition $(T,C,L)$. For each arc $e$ of $G'=G/T$ the tuple $(A(e), \omega(e), \eta(e), I(e))$.
A contractible closed walk $c$ in $G$.
    \Ensure $A(c)$ the signed area of $c$
    \State $dom \gets 0$
    \State $area \gets 0$
    \State $\alpha\gets 0$
    \For{each arc $e$ of $c$ not in $T$}
        \State $dom\gets dom + I(e)+\alpha \cdot \eta(e)$ \label{li:dom augment}
        \State $area\gets area + A(e)$ \label{li:area sum}
        \State $\alpha\gets \alpha+\omega(e)$ \label{li:alpha augment}
    \EndFor
    \State \Return $dom\cdot|F|-area$ \label{li:return}
\end{algorithmic}
\end{algorithm}
We let $c'$ be the walk in $G_L$ obtained by concatenation of the $c_e$'s for $e\in c/T$.
\begin{lemma}
\label{le:dom}
At the end of Algorithm~\ref{Algo area}, we have $dom=A_L(c')$.
\end{lemma}
\begin{proof}
  First note that if $c$ is contractible, so are $c/T$ and  $c'$.
For $\ell\in\mathbb N$ we denote by $c'(\ell)$ the sub-walk obtained by the concatenation of the walks $c_e$ corresponding to the $\ell$ first arcs of $c/T$. Let $v_0$ be a lift of $v$ and let $\widetilde c'$ be the lift of $c'$ starting from $v_0$. Recall the co-boundaries $\omega, \eta$ computed with respect to $G_L$. We consider the potential $\alpha_\omega$ such that $\omega$ derives from $\alpha_\omega$ and $\alpha_\omega(v_0)=0$. We denote by $dom_\ell$ and $\alpha_\ell$ the respective values of $dom$ and $\alpha$ in Algorithm~\ref{Algo area} at the beginning of the $\ell$-th iteration of the for loop (hence $\alpha_0=0$ and $dom_0=0$), ignoring the arcs of $c$ in $T$. Let $v_\ell$ be the $\ell+1$-th vertex of the lift of $c/T$ starting from $v_0$.
We prove by induction on $\ell$ that
\
\begin{align}\label{eq:algo1-induction}
  dom_\ell=\sum_{xy\in \widetilde{c'(\ell)}}\frac{\alpha_\omega(x)+\alpha_\omega(y)}{2}\eta(p(xy))\quad \text{and} \quad \alpha_\ell=\alpha_\omega(v_\ell)
\end{align}
for every $\ell$ smaller than or equal to the length of $c/T$.
The formula is trivially true for $\ell=0$.
For the inductive step, assume~\eqref{eq:algo1-induction} holds for some $\ell$ strictly smaller than the length of $c/T$. Then Line~\ref{li:dom augment} implies $dom_{\ell+1}=dom_\ell+I(e)+\alpha_\ell\cdot \eta(e)$, where $e=p(v_\ell v_{\ell+1})$ and where, by the induction hypothesis, $\alpha_\ell=\alpha_\omega(v_\ell)$ and
\[I(e) = \sum_{xy\in \widetilde{c_e}}\frac{\alpha'_\omega(x)+\alpha'_\omega(y)}{2}\eta(p(xy))
\]
for the potential $\alpha'_\omega=\alpha_\omega-\alpha_\omega(v_\ell)$ of $\omega$ satisfying $\alpha'_\omega(v_\ell)=0$. We infer
\begin{align*}
  I(e)+\alpha_\ell\cdot\eta(e) &=\sum_{xy\in \widetilde{c_e}}\frac{\alpha_\omega(x)+\alpha_\omega(y)}{2}\eta(p(xy)) -\sum_{xy\in \widetilde{c_e}}\alpha_\omega(v_\ell) \eta(p(xy))+\alpha_\omega(v_\ell)\eta(p(xy))\\
                               &=\sum_{xy\in \widetilde{c_e}}\frac{\alpha_\omega(x)+\alpha_\omega(y)}{2}\eta(p(xy)) - \alpha_\omega(v_\ell)\sum_{xy\in \widetilde{c_e}\cup \overline{e}}\eta(p(xy))\\
                               &=\sum_{xy\in \widetilde{c_e}}\frac{\alpha_\omega(x)+\alpha_\omega(y)}{2}\eta(p(xy))
\end{align*}
where we used that $\eta$ is a co-boundary for the last equality. It follows that
\begin{align}\label{eq:dom}
  dom_{\ell+1}=dom_\ell+\sum_{xy\in \widetilde{c_e}}\frac{\alpha_\omega(x)+\alpha_\omega(y)}{2}\eta(p(xy)) = \sum_{\widetilde{c'(\ell+1)}} \frac{\alpha_\omega(x)+\alpha_\omega(y)}{2}\eta(p(xy)).
\end{align}
Moreover, we have by Line~\ref{li:alpha augment} that $\alpha_{\ell+1}=\alpha_\ell+\omega(e)$. We thus get from the induction hypothesis and the definition of $\omega(e)$:
\[\alpha_{\ell+1}=\alpha_\omega(v_\ell)+\sum_{e'\in c_e}\omega(e')=\alpha_\omega(v_{\ell+1}),
\]
where the last equality follows from the definition of a potential. Together with~\eqref{eq:dom}, this confirms that~\eqref{eq:algo1-induction} holds at order $\ell+1$. By induction, the variable $dom$ thus contains 
$\sum_{xy\in \widetilde{c'}}\frac{\alpha_\omega(x)+\alpha_\omega(y)}{2}\eta(p(xy))$ at the end of Algorithm~\ref{Algo area}.
The lemma now follows from Proposition~\ref{alpha}.
\end{proof}

\begin{proposition}
\label{pr:correct2}
Algorithm~\ref{Algo area} computes the signed area $A(c)$ of $\widetilde c$.
\end{proposition}

\begin{proof}
  Recall from~\eqref{eq:contract-T} that $A(c)=A'(c/T)$. We can view any closed walk of $\widetilde G_L$ as a walk in $G_L$ supported by edges of $L$. With respect to such a walk, any face of $\widetilde G'$ has the same winding number as the face of $\widetilde G_L$ that contains it. Since every face of $\widetilde G_L$ contains $|F|$ faces of $\widetilde G'$, we deduce that
  \begin{align}\label{eq:unreduce}
    A'(c')=A_L(c')\cdot |F|,  
  \end{align}
  where $c'$ en $|F|$ are defined as in Lemma~\ref{le:dom} and Algorithm~\ref{Algo area}.

  Suppose exactly one arc $e$ of $c/T$ is in $\overset{\leftrightarrow}{C}$. Then $c'$  is obtained from $c/T$ by replacing $e$ with $c_e$. Denote by $F_e$ the set of faces of $\widetilde G'$ enclosed by any lift of the simple loop formed by $c_e$ concatenated with the arc opposite to $e$. Recall that $A(e)$ is the number of faces in $F_e$. Let $\widetilde{c/T}$ and  $\widetilde{c'}$ be lifts of respectively $c/T$ and $c'$, starting from the same vertex in $\widetilde G'$. Clearly, for each face $f$ of $\widetilde G'$ we have $wind(\widetilde{c/T},f)=wind(\widetilde{c'},f)-1$ if $f\in F_e$, and  $wind(\widetilde{c/T},f)=wind(\widetilde{c'},f)$ otherwise. Hence, $A'(c/T)=A'(c')-A(e)$.
  By an easy induction on the number of arcs of $c/T$ in $\overset{\leftrightarrow}{C}$, we deduce that $A'(c/T)=A'(c')-\sum_{e\in c/T}A(e)$ for any contractible closed walk $c$. So, by~\eqref{eq:contract-T} and~\eqref{eq:unreduce}, we get
  \[ A(c) = A'(c/T)=A_{L}(c')\cdot|F|-\sum_{e\in c/T}A(e).
  \]
  On the other hand, we easily infer from Line~\ref{li:area sum} of Algorithm~\ref{Algo area} that $area$ is equal to $\sum_{e\in c/T}A(e)$ at the end of the algorithm. Moreover, we have by Lemma~\ref{le:dom} that $dom = A_L(c')$ at the end of the algorithm. We infer that the expression returned by Algorithm~\ref{Algo area} at Line~\ref{li:return} is indeed $A(c)$.
\end{proof}
We now analyse the complexity of the area computation.
\begin{lemma}
\label{le:pre-computation}
The quantities $A(e),\omega(e),\eta(e),I(e)$ can be computed in $O(n)$ time where $n$ is the number of edges of $G$.
\end{lemma}
\begin{proof}
We first compute a tree-cotree decomposition of $G$ in $O(n)$ time~\cite{Epp02}. From this decomposition, we obtain representations of $G'$ and $G_L$ as rotation systems in linear time. Let $(e_1\ldots e_{4g})$ be the facial walk of $G_L$.
The computation of $\omega$ and $\eta$ in $G_L$ requires finding four indices $1\le i<j<k<\ell\le 4g$ such that $e_k=\overline{e_i}$ and $e_j=\overline{e_\ell}$. This can be done in $O(g)$ time. 
Consider the vertices $v_0,v_1,\ldots, v_{4g-1},v_0$ of a lift of the facial walk of $G_L$ where $p(v_{i-1}v_{i})=e_{i}$. In order to compute $\omega(e),\eta(e), I(e)$ for $e\in\overset{\leftrightarrow}{C}$, we first compute and store in $O(g)$ time all the partial sums $\alpha(v_i):=\sum_{j=1}^{i}\omega(e_j)$ and $\beta(v_i):=\sum_{j=1}^{i}\eta(e_j)$. Note that $\alpha$ coincides with the potential of $\omega$ that cancels at $v_0$. We then compute and store in $O(g)$ time all the partial sums $I_i=\sum_{j=1}^i\frac{\alpha(v_j)+\alpha(v_{j+1})}{2}\eta(e_j)$. For each arc $e$ of $G'$ with lift $v_iv_j$ (when $e$ is a loop edge we may have $j=i$ or $j=i+4g$ depending on the orientation of the loop) we can now compute $\omega(e)=\alpha(v_j)-\alpha(v_i)$ and $\eta(e)=\beta(v_j)-\beta(v_i)$ in constant time per arc. Moreover, we can compute $I(e)$ in constant time thanks to the formulas
\begin{align*}
    I(e)& = I_j-I_i+\alpha(v_i)(\beta(v_j)-\beta(v_i)) & \text{if } i\le j\\
    I(e)& = I_{4g}+I_j-I_i+\alpha(v_i)(\beta(v_j)-\beta(v_i)) & \text{if }  i>j
\end{align*}
By construction, the spanning cotree $C^*$ is the graph dual to the combinatorial surface $G'$ cut along $L$. This combinatorial surface is topologically a disk. Each arc $e\in \overset{\leftrightarrow}C$ cuts this disk into two parts. The combinatorial area $A(e)$ of the right part is thus equal to the number of vertices of the corresponding component of $C^*-e^*$. 
It follows that the $A(e)'s$ may be obtained by computing the sizes of the subtrees of $C^*$ in total linear time, thus ending the proof of the lemma.
\end{proof}
Thanks to Lemma~\ref{le:pre-computation} and Algorithm~\ref{Algo area}, we obtain:
\begin{proposition}
\label{prop:area}
After a pre-computation in $O(n)$ time, Algorithm~\ref{Algo area} computes the signed area of a walk $c$ in $G$ in time linear in the length of $c$.
\end{proposition}

\section{Homotopy and Simplicity Test}\label{sec:simplicity-test}
In this section we build on the works~\cite{LR12,EW13} that provide linear time contractibility test and linear time free homotopy test for walks in graphs embedded in surfaces. As a direct consequence of their construction, we provide an online algorithm for geodesic tightening of walks, see Proposition~\ref{prop:Q-pcan}. Additionally, we consider the problem of testing whether a walk in a graph cellularly embedded on a surface has a simple lift in the universal covering of that surface (see Theorem~\ref{th:simplicity}). These refinements are crucial in our efficient detection of balanced bigons in Section~\ref{sec:kernel-minor-algo}. Note that the free homotopy test is also used in Section~\ref{ssec:spectrum-equality-test} when considering homotopy of systems of curves.

\begin{theorem}
\label{th:simplicity}
Let $(S,G)$ be a combinatorial surface with $n$ edges and genus $g \ge 2$. After a precomputation in $O(n)$ time, we can test whether a walk $c$ in $G$ of length $\ell$ has a simple lift in the universal covering $\widetilde S$ in $O(\ell \log(\ell+g))$ time.
\end{theorem}

Let us mention that for the special case of the $1\times1$ grid in genus $g=1$, Brlek, Koskas and Proven\c{c}al~\cite{bkp-altasafdpi-2009} provides an optimal $O(\ell)$ time algorithm for this problem.
The simplicity of a lift is intimately related to homotopy. Namely, the lift of a walk $c$ of length $\ell$ in $G$ to the universal covering $\widetilde S$ is simple if and only if there does not exist a pair of indices $1 \le i < j \le \ell$ such that the lift of the subwalk $c[i\dots j]$ of $c$ between its $i$th and $j$th vertices is closed in $\widetilde S$, or equivalently such that the walk $c[i\dots j]$ is closed and contractible in $S$.

%\begin{theorem}\label{th:online-homotopy}
%Let $G$ be a cellularly embedded graph with $n$ edges on a genus $g \ge 2$ closed surface. After a precomputation in $O(n)$ time, for a walk $W$ in $G$ of length $\ell$ and any family $W_1\subset W_2\ldots \subset W_\ell$ of strictly increasing subwalk of $W$, testing if there is $1\le i\le \ell$ such that $W_i$ is contractible may be done in $O(\ell)$.
%\end{theorem}

\subsection{System of quads and Homotopy Tests}
Even though it is hidden in the statement of Theorem~\ref{th:simplicity}, the construction of~\cite{LR12,EW13} that we will use to prove involves a simplified embedded graph called a \define{system of quads}, see Section~3 in~\cite{EW13}. Systems of quads are closely related to the reduced graphs considered in the previous section. Given a combinatorial surface $(S,G)$ with $n$ edges, a system of quads $Q$ (cellularly embedded in $S$) can be computed in $O(n)$ time from a tree-co-tree decomposition of $(S,G)$ and comes with a map $q$ from the edges in $G$ to the walks in $Q$ whose main properties are recalled in the following statement.
\begin{proposition} \label{prop:Q-q}
Let $(S,G)$ be a combinatorial surface of genus $g \ge 2$. There exists a system of quads $Q$ on $S$ and a map $q$ from the edges of $G$ to walks in $Q$ such that
\begin{itemize}
\item $q$ is compatible with the concatenation: if $c$ and $c'$ are two walks in $G$ that can be concatenated, then $q(c \cdot c') = q(c) \cdot q(c')$. In the case where $c'=e$ is a single edge, then $q(e)$ has length zero or two.  The collection of all $q(e)$ can be precomputed in $O(n)$ time and stored in a table of linear size. As a consequence, $q(c)$ can be computed in $O(|c|)$ time. 
\item Let $c$ and $c'$ be two walks in $G$ with the same endpoints. Then $c$ and $c'$ are homotopic (with fixed endpoints in $G$) if and only if $q(c^{-1}.c')$ is (freely) homotopic to a trivial walk.
\end{itemize}
\end{proposition}
Proposition~\ref{prop:Q-q} lies at the start of the construction in~\cite{LR12}. It allows to perform homotopy tests in $Q$ rather than in $G$.
The homotopy test in the system of quads $Q$ is done by computing a unique geodesic representative for each homotopy class of walks. In fact, as we only need to test contractibility of walks, in this paper we only need to compute a geodesic representative. Recall from~\cite{EW13, DL19} that a walk is geodesic if and only if it does not contain spurs or brackets.
In order to achieve linear time complexity one can encode paths in $Q$ as sequences of turns in run-length encoding form, see Section~4.1 in~\cite{EW13}. In this encoding, a spur corresponds to a turn $0$ and a bracket of length $k$ to the compressed sequence of turns $(1,2^k,1)$ or $(\overline{1}, \overline{2}^k, \overline{1})$.
%The main algorithm of~\cite{EW13} can be summarized as follows.
%\begin{proposition}\label{prop:Q-can}
%Let $Q$ be a system of quads. There exists a map $\can$ transforming walks of even lengths in $Q$ into canonical ones such that
%\begin{itemize}
%\item Given a closed walk $c$ of length $2\ell$, its canonical form $\can(c)$ can be computed in $O(\ell)$ time. Moreover, $c$ and $\can(c)$ are freely homotopic on $S$.
%\item Two closed walks $c$ and $c'$ on $Q$ are freely homotopic if and only if $\can(c)$ is a cyclic permutation of $\can(c')$.
%\end{itemize}
%\end{proposition}
We now provide a refinement of the results in~\cite{LR12,EW13} that follows from their techniques. 
\begin{proposition}\label{prop:Q-pcan}
Let $w$ be a geodesic walk in $Q$ represented as a sequence of turns and the first and last edges. Given an edge $e$ in $Q$ such that $w\cdot e$ is a walk (respectively $e\cdot w$), we can compute in $O(1)$ time the sequence of turn with the first and the last edge of a geodesic $w'$ homotopic to $w\cdot e$ (resp. $e\cdot w$) as paths with fixed end points.
%Let $Q$ be a system of quads. There exists a map $pcan$ on walks in $Q$ such that
%\begin{itemize}
%    \item Given a geodesic walk $w=pcan(w)$ represented by a sequence of turns $s$ with starting edge $f$  and an edge $e$ in $Q$,  then $s$ can be modified to represent $pcan(e \cdot s)$ or $pcan(s \cdot e)$ in $O(1)$ time by updating respectively its prefix or its suffix.
%    \item A pair of sequence of turns $s$ and $s'$ are homotopic (with fixed endpoints) if and only if $pcan(s) = pcan(s')$. \todo{NO!, you need to start with the same edge.}
%\end{itemize}
\end{proposition}
%Note that in the statement, in order for the updates to be done in constant time an appropriate data structure such as double ended queues must be used for the sequences of turns. Alternatively, the algorithm can provide the $O(1)$ size list of things to be updated.

\begin{proof}[Sketch]
%Recall from~\cite{LR12,EW13} that the construction of $can$ is done by first finding a geodesic representative (see Section~4.3 in~\cite{EW13}) and then a canonical choice among such geodesics (see Section~5.1 in~\cite{EW13}). Both steps can be performed on walks in $Q$ rather than cycles and the second item of the statement follows.
The case of appending at the end or the beginning of $w$ are symmetric and we consider only the former. Let $t$ be the turn between the last edge of $w$ and $e$. Adding the turn $t$ to the sequence of turns of $w$ does not necessarily result in a geodesic. Though, it can only create a single spur (of length 2) or a bracket (that can be arbitrarily long, but in the run-length encoding it has constant size). Removing either of them we obtain a sequence of turns of a geodesic representative for $w \cdot e$ by~\cite[Cor 5.1]{GS90}. Moreover, we can update the last edge of this geodesic in constant time.
\end{proof}

We now provide the main construction at the heart of the simplicity test of Theorem~\ref{th:simplicity}. Given a system of quads $Q$ in a closed oriented surface $S$ we show how to explore the preimage $\widetilde Q$ of $Q$ in the universal covering $\widetilde S$ of $S$. Let us recall that $Q$ and hence $\widetilde{Q}$ are bipartite graphs. Hence, given a root vertex $v_0$ in $\widetilde{Q}$, for each edge $e$ of $\widetilde{Q}$ the distances from $v_0$ to the endpoints of $e$ must differ by $1$. As a consequence, there exists a canonical orientation of the edges of $\widetilde{Q}$ directed towards the root. From any given vertex $v$ of $\widetilde{Q}$, any oriented path from $v$ to $v_0$ is a geodesic. We note that it follows from~\cite{LR12} that the outgoing degree of any vertex $v$ of $\widetilde{Q}$ different from $v_0$ is either $1$ or $2$. In the case where $v$ has degree $2$ its two outgoing edges form the corner of a square of $\widetilde{Q}$.

Let $c$ be a walk in $Q$ and let $v_0$ be a vertex in $\widetilde Q$ that is a lift of the start of the walk. We consider $v_0$ as a root in $\widetilde Q$ and consider its edges oriented accordingly. Let $v_1, \ldots, v_k$ be the other vertices of the lift $\widetilde c$ of $c$ in $\widetilde Q$ starting at $v_0$. We define the oriented graph $\widetilde Q(c)$ as the subgraph of $\widetilde Q$ induced by the union of all geodesics from $v_0$ to some $v_k$. In other words $\widetilde{Q}(c)$ is the smallest star-shaped subset of $\widetilde{Q}$ containing $\widetilde{c}$. From the definition of $\widetilde{Q}(c)$, it follows directly that if $c$ and $c'$ are two walks in $Q$ that can be concatenated then $\widetilde{Q}(c) \subset \widetilde{Q}(c \cdot c')$. The main result about $\widetilde{Q}(c)$ is that it can be built in almost linear time.

\begin{proposition}\label{pr:tilde-Q-update}
Let $Q$ be a system of quads embedded in a closed oriented surface $S$ of genus $g$. Let $p: \widetilde{Q} \to Q$ be its universal cover and $v_0$ a root in $\widetilde{Q}$. Let $c$ be a walk of length $\ell$ in $Q$ starting at $p(v_0)$. Then calling Algorithm~\ref{Algo edge insertion} successively on the edges of $c$ updates $\{v_0\}$ as $\widetilde{Q}(c)$ in $O(\ell\ \log g)$ total time.
\end{proposition}

Before making the proof of this proposition, we specify the data structure and primitives that we use. The incremental update in Algorithm~\ref{Algo edge insertion} works for more general star-shaped subgraphs $H$ of $\widetilde{Q}$ and we work in this context.
The vertices of $H \subset \widetilde{Q}$ are indexed by integers starting from $0$ representing the root to $n-1$ where $n$ is the number of vertices of $H$. Each edge is represented by a triple $(u,v,e)$ where $u,v \in [0,n-1]$ represent vertices of $H$ and $e$ is an edge in $Q$. As we explain below, we can assume the following primitives to be available
\begin{itemize}
\item \texttt{add\_vertex}$(H)$: add a vertex to $H$ and return its index in $O(1)$ time.
\item \texttt{rotate}$(H, u, e, t)$: return the edge in $E(Q)$ that is obtained by rotating around $p(u)$ by a turn $t$.
\item \texttt{contains\_edge}$(H, u, e)$: return whether $H$ contains the lift of the edge $e \in E(Q)$ from $u \in V(H)$ in $O(\log g)$ time.
\item \texttt{opposite\_vertex}$(H, u, e)$: return the vertex which is the other end of the lift of $e \in E(Q)$ starting at $u \in V(H)$ in $O(\log g)$ time.
\item \texttt{add\_edge}$(H, u, v, e)$: add an edge to $H$ from the vertex $u \in V(H)$ to the vertex $v \in V(H)$ that projects to the edge $e \in E(Q)$ in $O(\log g)$ time.
\item \texttt{turn}$(H, u, e)$: return the pair $(t, e')$ made of the turn $t$ with smallest absolute value between the lift of the edge $e \in E(Q)$ based at $u \in V(H)\setminus\{0\}$ and the (at most 2) outgoing edges at $u$ in $O(1)$ time. Also return the edge $e'$ of $Q$ which realizes this turn. If $u$ is the root, then return a special value ($+\infty$, \texttt{None}).
\end{itemize}
We now explain how to implement our data structure.
We store the edges adjacent to $v \in V(H)$ in a self-balancing binary search tree so that edge search and addition is done in $O(\log g)$ time. We also store the outgoing edges of $H$ for each vertex. Moreover, we index the edges of $Q$ cyclically around its two vertices so that the turn between two edges around a vertex of $Q$ can be obtained in constant time from their indices. This allows \texttt{turn}$(H, u, e)$ and \texttt{rotate}$(H, u, e, t)$ to run in $O(1)$ time .

\begin{algorithm}[ht!]
\caption{Edge insertion with star-shaped saturation}
\label{Algo edge insertion}
\begin{algorithmic}[1]
  \Require A non-empty star-shaped subgraph $H$ of $\widetilde{Q}$, a vertex $u \in H$ and an edge $e \in Q$ that has a lift $\widetilde{e}$ in $\widetilde{Q}$ starting at $u$.
    \Ensure Update $H$ with $\widetilde{e}$ and all geodesics containing its other endpoint.

    \If{\texttt{has\_edge}$(H, u, e)$} \label{algo step:edge present start}
        \State \Return
    \EndIf \label{algo step:edge present end}
    \State $(t,e_1) \gets $ \texttt{turn}$(H, u, e)$ \label{algo step:turn definition}
    \State $v \gets$ \texttt{new\_vertex}$(H)$
    \If{$t = \pm 1$} \label{algo step:large turn start}
        \State $u_1\gets$ \texttt{opposite\_vertex}$(H,u,e_1)$ 
        \State $e_2\gets$ \texttt{rotate}$(H,u_1,-t)$
         \State Recursively call Algorithm~\ref{Algo edge insertion} with $H,u_1,e_2$ \label{algo step:recursion}
         \State $u_2\gets$ \texttt{opposite\_vertex}$(H,u_1,e_2)$
         \State $e_3\gets$ \texttt{rotate}$(H,u_2,e_2,-t)$
         \State \texttt{add\_edge}$(H,v,u_2,e_3)$
    \EndIf \label{algo step:large turn end}
    \State \texttt{add\_edge}$(H, v, u, e)$ \label{algo step:add edge e}
    \State \Return $H$
\end{algorithmic}
\end{algorithm}

\begin{figure}[ht!]
    \centering
    \includegraphics[width=7cm]{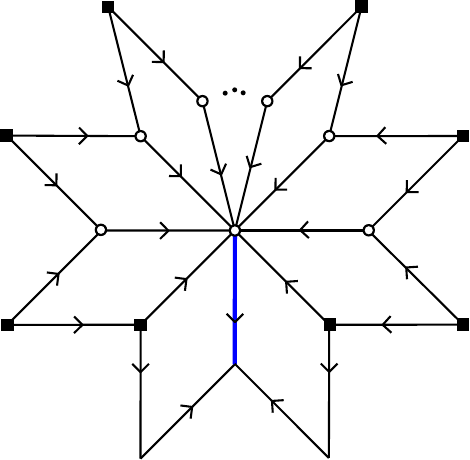}
    \includegraphics[width=7cm]{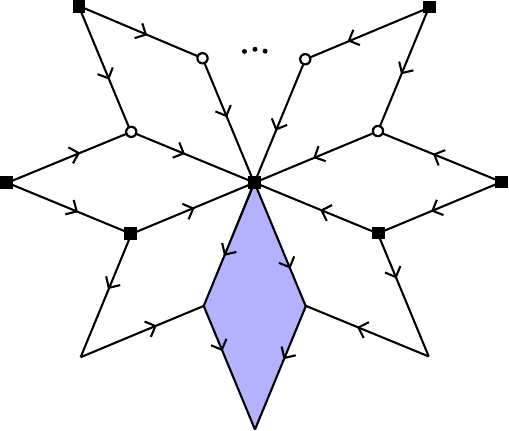}
    \caption{The two possible orientations of the edges in the neighborhood of a non-root vertex in $\widetilde{Q}$. The black squared and white circle vertices have respectively outgoing degree two and one. The outgoing degrees of the unmarked vertices are not determined by the outgoing degree of the central vertex.}
    \label{fig:Qtilde-neighborhoods}
\end{figure}

\begin{proof}
We start from $H_0=\{v_0\}$ and define $H_i, v_i$ the state of $H$ after calling Algorithm~\ref{Algo edge insertion} on $H_{i-1}, v_{i-1}$ and $c[i]$ where $c[i]$ is the $i$-th edge of $c$ and $v_{i}$ is \texttt{opposite\_vertex}$(H_i,v_{i-1},c[i])$.
We first prove the correctness of Algorithm~\ref{Algo edge insertion}: $H_i$ is the star-shaped closure of the prefix of length $i$ of the lift of $c$ starting at $v_0$. We proceed by induction.

The case of $u=v_{i-1}$ being the root is elementary and correctly treated by either the presence detection at Lines~\ref{algo step:edge present start}-\ref{algo step:edge present end} or the edge addition at Line~\ref{algo step:add edge e}. There is no need to saturate in that case. From now on we assume that the input $u$ is not a root vertex and that the lift $\widetilde{e}$ of $e=c[i]$ at $u$ is not present in $H_{i-1}$. Let us note that this edge is added to $H_{i}$ at the very end of the algorithm at Line~\ref{algo step:add edge e}.
The neighborhood of $u$ in $\widetilde{Q}$ can be of two different kinds with respect to edge orientations as can be seen in Figure~\ref{fig:Qtilde-neighborhoods}. Let $t$ and $e_1$ denote as in Line~\ref{algo step:turn definition} the turn between the outgoing edges at $u$ and $\widetilde{e} \in E(\widetilde{Q})$ and the corresponding outgoing edge. Then $t = \pm 1$ if and only if the outgoing degree of $v$ at the other end of $e$ is two, see Figure~\ref{fig:Qtilde-neighborhoods}.
\begin{itemize}
\item If the turn is large, that is $t \not= \pm 1$, then we only need to add the vertex $v$ and the edge $\widetilde{e}$ to $H$.
\item If the turn is small, that is $t = \pm 1$, we refer to Figure~\ref{fig:insert_edge} and Lines~\ref{algo step:large turn start}-\ref{algo step:large turn end}. The algorithm adds to $H$ the whole square adjacent to $e$ after having called recursively the edge insertion with $(u_1, e_2)$. By induction, one sees that the updated graph $H$ is indeed star-shaped. This concludes the correctness.
\end{itemize}
\begin{figure}[ht!]
    \centering
    \includegraphics[width=4cm]{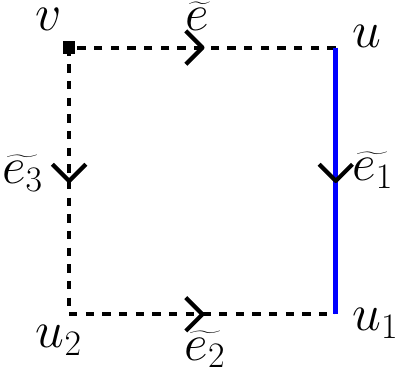}
    \caption{Illustration of the step at Lines~\ref{algo step:large turn start}:\ref{algo step:large turn end} of Algorithm~\ref{Algo edge insertion}.}
    \label{fig:insert_edge}
\end{figure}

We now prove the complexity upper bound. Each call to Algorithm~\ref{Algo edge insertion} ignoring the recursion at Line~\ref{algo step:recursion} either returns immediately (if the edge $\widetilde{e}$) is present or adds between 1 and 3 edges to $H$. In particular, including the recursion step, the call to Algorithm~\ref{Algo edge insertion} takes a time proportional to the number of edges added to $H$ times $\log g$ (recall that our primitives for modifying the graph $H$ takes either $O(1)$ or $O(\log g)$ time).

The complexity upper bound is now obtained by showing that the size of $\widetilde{Q}(c)$ is $O(\ell)$. By Theorem~\ref{th:isoperimetric-quad-system}, it reduces to prove that the perimeter of $H$ is bounded by $2\ell$. Indeed, the perimeter of $H_i$ has at most $2$ more edges compared to $H_{i-1}$. This shows that the number of edges in $\widetilde{Q}(c)$ is $5\ell$. Thus the complexity to compute $\widetilde{Q}(c)$ is $O(\ell\log g)$.
\end{proof}
 
\subsection{Simplicity}
The goal of this section is to prove Theorem~\ref{th:simplicity}.

Let $(S,G)$ be a combinatorial surface of genus $g \ge 2$. We first prove a lemma that shows how to use the construction from Proposition~\ref{pr:tilde-Q-update} to build coordinates for $\widetilde G$ in the universal covering $\widetilde S$. Let $c$ be a walk of length $\ell$ in $G$ with starting point $v=v(0)$. For any $0\le k\le \ell$, we denote by $v(k)$ the $k$-th vertex of $c$, and by $c(k)$ the restriction of $c$ to its $k$ first edges. Choose $w(0)$ in $\widetilde{Q}$ above the first vertex $q(v)$ of $q(c)$. Let also $w(k)$ be the endpoint of the lift of $q(c(k))$ starting at $w(0)$. We finally consider the map $\psi:k\mapsto (v(k), w(k))$.
\begin{lemma}
\label{le:simplicity-equivalence}
$\widetilde c$ is simple if and only if $\psi$ is injective.
\end{lemma}

\begin{proof}
Assume $\widetilde c$ is simple. Then for every $k\neq k'$ such that $v(k)=v(k')$, the paths $c(k)$ and $c(k')$ are not homotopic. Thus $q(c(k))$ and $q(c(k'))$ are not homotopic by Proposition~\ref{prop:Q-q}. It follows that their lifts from $w(0)$ ends at distinct points, i.e. $w(k)\neq w(k')$. Hence, for every $k\neq k'$, $\psi(k)\neq\psi(k')$.

Assume $\widetilde{c}$ is not simple. Then there are $k\neq k'$ such that $\widetilde c(k)$ and $\widetilde c(k')$ end at the same vertex of $\widetilde G$. Thus $c(k)$ and $c(k')$ are homotopic. So $q(c(k))$ and $q(c(k'))$ are homotopic and $w(k)=w(k')$. Furthermore $v(k)=v(k')$. It follows that $\psi(k)=\psi(k')$ implying that $\psi$ is not injective.
\end{proof}

\begin{proof}[Proof of Theorem~\ref{th:simplicity}]
Thanks to Lemma~\ref{le:simplicity-equivalence}, the non-simplicity of $c$ is equivalent to the existence of two indices $k$ and $k'$ such that $\psi(k) = \psi(k')$. Via Proposition~\ref{prop:Q-q} and Proposition~\ref{pr:tilde-Q-update}, we compute $\widetilde Q(q(c))$ and $w(k)$ for $1\le k\le \ell$ in $O(\ell \log g)$ time. The intersection test is now reduced to the existence of a collision among the $\ell$ pairs $(v(k), w(k)) \in V(G) \times V(\widetilde{Q}(q(c)))$. This can be achieved in $O(\ell\, \log \ell)$ time.

The overall cost is hence dominated by the collision test and this concludes the proof.
\end{proof}

\section{Minor Kernel Algorithm}\label{sec:kernel-minor-algo}

In this section we prove Theorem~\ref{thm:minor-kernel}, which we repeat for convenience.

\medskip\noindent
{\sc T{\footnotesize HEOREM}~\ref{thm:minor-kernel}.}
{\it 
Let $G$ be a graph with $n$ edges cellularly embedded in a closed oriented surface $S$ of genus $g\ge 2$. A minor kernel of $G$ can be computed in $O(n^3 \log n)$ time.
}

\medskip\noindent
Thanks to the correspondence between vertex metric and cross metric, the theorem is equivalent to the following:

\begin{theorem}
\label{th:Curve-minimal-algo}
Let $S$ be a closed oriented surface of genus $g\ge 2$ and let $\mathcal C$ be a filling system of closed curves in $S$ with $n$ vertices. Then, a smoothing minimal system for $\mathcal C$ can be computed in $O(n^3 \log n)$ time.
\end{theorem}
The proof of Theorem~\ref{th:Curve-minimal-algo} essentially relies on the existence of an algorithm to detect minimal bigons in $O(n^2 \log n)$ time. Thanks to Proposition~\ref{Smooth spectrum}, we can safely smooth any of the two corners of such bigons without changing the length spectrum. We repeat this operation, also taking care of empty monogons, until there is no more empty monogon or minimal bigon. By Proposition~\ref{prop:minimal} we must have a tight system of primitive curves, hence a minor kernel. Now, the search for minimal bigons proceeds as follows. Thanks to Theorem~\ref{thm:a-minimal-bigon-has-linear-length}, we can restrict the search to bigons of linear length. The boundary of such bigons projects on $S$ to the concatenation of two homotopic subpaths of some curves in $\mathcal C$. We thus search among the curves in $\mathcal C$ for pairs of homotopic subpaths of (the same) linear length. We then check whether their concatenation lifts to a bigon in $\widetilde S$, i.e. to a simple closed walk. This can be done efficiently according to the simplicity test in Section~\ref{sec:simplicity-test}. In case the test succeeds, we compute the area of the bigon as explained in Section~\ref{sec:computation-area}. It remains to return a bigon with minimal area among the selected bigons.

Algorithm~\ref{Algo} below summarizes some of the above discussion. More precisely, the algorithm returns True if a given system $\mathcal{C}$ of closed curves, assumed without empty monogon, is a tight system of primitive curves. If not, it returns a minimal bigon. To this end, we search for a bigon by looking at every of the four sectors incident to every vertex of $\mathcal C$ (viewed as a 4-regular graph), trying to extend this sector to a bigon of length at most $8n$. If no such bigon is found, then we can certify that the system is a tight system of primitive curves and the algorithm stops. Otherwise, we may select a minimal bigon among the $O(n)$ bigons found (at most one bigon per sector). In practice a sector is fully determined by a vertex $v$ of $\mathcal C$ and the two oriented curves $C_i^\varepsilon$ and $C_j^\eta$ bounding the sector with $\varepsilon\in\{-1,1\}$ and $\eta\in\{-1,1\}$. In the following algorithm, we denote by $C^{\varepsilon}_i(v,\ell)$ the walk of length $\ell$ starting at $v$ along $C_i$ in the direction given by $\varepsilon$.

\begin{algorithm}[ht!]
  \caption{minimal bigon search}
  \label{Algo}
\begin{algorithmic}[1]
    \Require a system of closed curves $\mathcal C$ on $S$ in general position with no empty monogon
    \Ensure \texttt{True} if $\mathcal C$ is a tight system of primitive curves. Otherwise return a corner of a minimal bigon.
    \State $A_0\gets \infty$, $Sector\gets$\texttt{None}\Comment{Initialize area and sector}
    \For{each sector $s$ with corner vertex $v$ bounded by curves $C^{\varepsilon}_i$ and $C^\eta_j$\label{li:Forloop}}
        \State $\ell\gets0$
        \State $IsABalancedBigon \gets$ \texttt{None}
        \Repeat{\label{li:Repeatloop}}
            \State $\ell\gets\ell+1$
            \If{$C^{\varepsilon}_i(v,\ell)$ and $C^{\eta}_j(v,\ell)$ are homotopic\label{li:homotopy}}
                \If{$C^{\varepsilon}_i(v,\ell) \cdot (C^{\eta}_j(v,\ell))^{-1}$ lifts to $\widetilde{S}$ as a simple curve}\label{li:simplicity}
                    \State IsABalancedBigon $\gets$ \texttt{True}\label{li:set-IsABalancedBigon-True}
                \Else
                    \State IsABalancedBigon $\gets$ \texttt{False}\label{li:set-IsABalancedBigon-False}
                \EndIf
            \EndIf
        \Until{$\ell\ge4n$ or $IsABalancedBigon \not=$ \texttt{None} \label{li:ERepeatloop}}
        \If{$IsABalancedBigon =$ \texttt{True}}
            \State Compute the area $A$ of the found bigon\label{li:Area}
            \If{$A_0> A$}
                \State $A_0,Sector\gets A,s$
            \EndIf
        \EndIf
    \EndFor \label{li:endFor}
    \If{$Sector=$\texttt{ None}} \Comment{No bigon was found}
        \State \Return True
    \Else
        \State \Return $Sector$ \label{li:return-sector}
    \EndIf
\end{algorithmic}
\end{algorithm}

\begin{proposition}\label{prop:proof-algo-7.1}
Algorithm~\ref{Algo} correctly returns a corner of a minimal bigon if there is one or detects if the system of closed curves is a tight system of primitive curves.
\end{proposition}

\begin{proof}
We use the notation of Algorithm~\ref{Algo}. We say that a bigon in $\widetilde{S}$ is \define{balanced} if its two sides have the same length. Note that a minimal bigon is balanced but the converse does not hold in general.
We first claim that for each sector $s$ in the FOR loop (Lines \ref{li:Forloop}-\ref{li:endFor}), the state of $IsBalancedBigon$ at the end of the repeat loop at Line~\ref{li:ERepeatloop} is
\begin{itemize}
    \item \texttt{True} if $C^{\varepsilon}_i(v,\ell)$ and $C^{\eta}_j(v,\ell)$ lift in $\widetilde{S}$ as the sides of a balanced bigon of length $2 \ell$ starting with sector $s$.
    \item \texttt{False} or \texttt{None} if there is no balanced bigon of length smaller than $8n$ starting with sector $s$.
\end{itemize}
Let us prove this claim. In the case $IsBalancedBigon$ is \texttt{True} or \texttt{False}, then this is because the code entered Line~\ref{li:set-IsABalancedBigon-True} or Line~\ref{li:set-IsABalancedBigon-False}. Furthermore, as soon as $IsBalancedBigon$ is set to \texttt{True} or \texttt{False}, the algorithm exits the repeat loop at line~\ref{li:ERepeatloop}. Let $U$ and $D$ be the lifts of $C^{\varepsilon}_i(v,\ell)$ and $C^{\eta}_j(v,\ell)$, respectively, starting at the same  vertex above $v$. Then, $U$ and $D$ end at the same point in $\widetilde{S}$ since the test at Line~\ref{li:homotopy} succeeded. In the case $IsBalancedBigon$ is \texttt{True}, then $D \cdot U^{-1}$ is a simple closed curve in $\widetilde{S}$ as the test at Line~\ref{li:simplicity} succeeded. It proves that $D$ and $U$ form the sides of a bigon as, by the Jordan curve theorem, a simple curve bounds a disk. This bigon is balanced of length $2 \ell \le 8n$ since the length of both curves is $\ell$ and $\ell \le 4n$. Now if $IsABalancedBigon$ is \texttt{False}, then $D \cdot U^{-1}$ is not simple since the test at Line~\ref{li:simplicity} failed. This could only happen because either $\widetilde{D}$ or $\widetilde{U}$ has a self-intersection, or because $\widetilde{D}$ and $\widetilde{U}$ intersect more than at their endpoints. In both situations there exists no balanced bigon starting with sector $s$: it could not be shorter than $\ell$ because the repeat loop checked all lengths from $1$, and it could not be longer because the boundary of a bigon must be a simple closed curve. Finally, if $IsABalancedBigon$ is \texttt{None} then it means that for any $\ell$ from $1$ to $4n$ the condition in line~\ref{li:homotopy} failed so that none of the lifts of the pairs $C^{\varepsilon}_i(v,\ell)$ and $C^{\eta}_j(v,\ell)$ end at the same point and could have formed a bigon. In that case, there is no balanced bigon starting from $s$ of length smaller than $8n$. This concludes the proof of the claim.

If $\mathcal{C}$ is a tight system of primitive curves, then $\widetilde{\mathcal C}$ has no bigon by Proposition~\ref{prop:minimal}. Hence Algorithm~\ref{Algo} fails to detect a balanced bigon and indeed returns that $\mathcal C$ is a tight system of primitive curves.

Assume now that $\mathcal{C}$ is not a tight system of primitive curves and has no empty monogon. Then due to Proposition~\ref{prop:minimal}, $\widetilde{\mathcal{C}}$ has a minimal bigon which has length at most $8n$ by Theorem~\ref{thm:a-minimal-bigon-has-linear-length}. Let $s$ be the sector corresponding to one of the corners of this minimal bigon. Then, as we already mentioned this bigon is balanced and will be detected by the repeat loop at lines~\ref{li:Repeatloop}-\ref{li:ERepeatloop} when entering the for loop with $s$ at Line~\ref{li:Forloop}.
The remaining part of the algorithm just selects among all balanced bigons of length at most $8n$  one with minimal area and $Sector$ is set to the sector of this bigon, call it $b$. We claim that $b$ is minimal. Otherwise $b$ would contain another minimal bigon with stricly smaller area. This other bigon, say $b'$ would be balanced and have length at most $8n$ by Theorem~\ref{thm:a-minimal-bigon-has-linear-length}. The bigon $b'$ would thus have been detected starting with another sector. This would however contradict that $b$ has minimal area among the found bigons. It follows that the corner stored in $Sector$ at Line~\ref{li:return-sector} is indeed the corner of a minimal bigon.
\end{proof}
 
\begin{proposition}\label{prop:complexity-algo-7.1}
Algorithm~\ref{Algo} can be implemented to run in $O(n^2 \log n)$ time where $n$ is the number of vertices of $\mathcal C$.
\end{proposition}

\begin{proof}
Note that the number of sectors is four times the number of vertices and thus linear in $n$.  We thus enter the for loop at Line~\ref{li:Forloop} a linear number of times. Let $\kappa(n)$ denote the complexity of executing the repeat loop from line~\ref{li:Repeatloop} to line~\ref{li:ERepeatloop} and let $\alpha(n)$ denote the complexity of computing the area of a bigon as in line~\ref{li:Area}. The total complexity of Algorithm~\ref{Algo} is $O\left(n\left(\kappa(n)+\alpha(n)\right)\right)$.

We first show that $\kappa(n) = O(n \log n)$.
In order to perform the homotopy test at line~\ref{li:homotopy}, we rely on Proposition~\ref{prop:Q-pcan}. More precisely, a system of quads $Q$ associated to $G$ and the map $q$ on walks is computed in $O(n)$ time according to Proposition~\ref{prop:Q-q}. Next, we maintain a geodesic representative $g(\ell)$ homotopic to the walk $q(C^{\varepsilon}_i(v,\ell)\cdot C^{\eta}_j(v,\ell)^{-1})$ in $Q$. Now $C^{\varepsilon}_i(v,\ell)$ and $C^{\eta}_j(v,\ell)$ are homotopic if and only if $g(l)$ is reduced to a vertex (in which case $C^{\varepsilon}_i(v,\ell)$ and $C^{\eta}_j(v,\ell)$ have the same endpoints). Proposition~\ref{prop:Q-pcan} provides that $g(\ell+1)$ is computed from $g(\ell)$ in constant time. Next, the fact that the simplicity test in line~\ref{li:simplicity} runs in $O(n \log n)$ time is a direct application of Theorem~\ref{th:simplicity}. This concludes the proof that $\kappa(n)=O(n \log n)$.

Moreover, Theorem~\ref{thm:linear-time-area-computation} states that we can compute the area of any bigon of length $l$ in $O(l)$ time. As the length of the bigon in line~\ref{li:Area} is at most $8n$, we deduce $\alpha(n)=O(n)$.

We conclude that Algorithm~\ref{Algo} takes $O(n^2 \log n)$ time.
\end{proof}

\begin{proof}[Proof of Theorem~\ref{th:Curve-minimal-algo}] Let $\mathcal C$ be a filling system of closed curves given as a combinatorial surface whose graph is 4-regular and contains $n$ vertices. 
We can search for empty monogons, i.e. for faces of degree $1$, in linear time simply by traversing the set of faces of the combinatorial surface. If there is an empty monogon, we smooth its corner and update the combinatorial surface accordingly in constant time. We may repeat this procedure until there is no more monogon. Using Algorithm~\ref{Algo}, we can then search for a minimal bigon and smooth one of its corners. This takes $O(n^2\log n)$ time by Propositions~\ref{prop:proof-algo-7.1} and~\ref{prop:complexity-algo-7.1}. 
We repeat this procedure, alternating between the smoothing of all empty monogons and the smoothing of a minimal bigon, until there is no more. As previously discussed this will provide a smoothing minimal system of curves. Since each smoothing removes a vertex of the combinatorial map, the total number of smoothings is bounded by $n$. It follows that the complexity of the above algorithm is $O(n^3\log n)$.
\end{proof}

\section{Applications}
\label{sec:applications}

\subsection{Testing Length Spectrum Equality}\label{ssec:spectrum-equality-test}

In this section, we provide an algorithm for the proof of Theorem~\ref{thm:H-G-equivalence}. With no additional cost we actually prove the following slightly stronger form.
\begin{theorem}\label{thm:H-H'-equivalence}
Let $(S,G)$ be a combinatorial surface of genus $g\ge 2$. Let $H$ and $H'$ be graph minors of $G$ given as lists of minor operations on the edges of $G$. We can decide whether $H$ and $H'$ have the same $\mu$-spectrum in $O(n^3 \log n)$ time.
\end{theorem}
Note that Theorem~\ref{thm:H-G-equivalence} reduces to Theorem~\ref{thm:H-H'-equivalence} when $H'=G$.
The overall cost of our algorithm for Theorem~\ref{thm:H-H'-equivalence} is dominated by the computation of minor kernels of $H$ and $H'$. The remaining part of the algorithm relies on the efficient comparison of systems of curves as stated in Theorem~\ref{thm:homotopy-system-of-curves} below. We say that two systems of curves $\mathcal C = (c_1, \ldots, c_k)$ and $\mathcal C' = (c'_1, \ldots, c'_{k'})$ are \emph{freely homotopic up to orientation} if $k=k'$ and there exists a permutation $\sigma$ of $\{1,\ldots,k\}$ such that for each $i \in \{1,\ldots,k\}$ either $c_i$ or $c^{-1}_i$ is freely homotopic to $c'_{\sigma(i)}$.
\begin{theorem}\label{thm:homotopy-system-of-curves}
Let $(S,G)$ be a combinatorial surface with $n$ edges and genus $g\ge 2$. Given two systems of closed walks in $G$ with total lengths $\ell$, we can decide whether they are freely homotopic up to orientation in $O(n + \ell\log g)$ time.
\end{theorem}
\begin{proof}
Let $\mathcal C = (c_1, \ldots, c_k)$ and $\mathcal C' = (c'_1, \ldots, c'_{k'})$ be the given curve systems. Our goal is to decide the existence of a permutation $\sigma$ such that $c_i$ or $c^{-1}_i$ is freely homotopic to $c'_{\sigma(j)}$.
We use the construction described in Section~\ref{sec:simplicity-test}. Recall that $Q$ is a system of quads associated to $G$ that comes with a map $q$ from the walks of $G$ to the walks of $Q$, and a map $can$ that gives a canonical representative in the free homotopy class of any closed walk of $Q$. After a precomputation in $O(n)$ time, for any closed walk $c$ in $G$, $can(q(c))$ may be computed in $O(|c|)$ time.
Hence, we can compute the canonical representatives $can(q(c_i))$ and $can(q(c'_j))$ of each walk $c_i$ or $c'_j$ in $O(\ell)$ total time.

As mentioned above, the walk $c_i$ or $c^{-1}_i$ is freely homotopic to $c'_j$ if and only if $can(q(c_i))$ or $can(q(c^{-1}_i))$ is a circular permutation of $can(q(c'_j))$. To solve this problem of comparison up to circular permutation, we order the canonical representatives as follows. We fix an arbitrary total order on the oriented edges of $Q$ so that walks are identified to words on $\{1,\ldots,8g\}$. Given such a word $w$, the computation of the lexicographically minimal word among its cyclic permutations is a standard problem that can be solved in linear time in the length $|w|$ (see~\cite{booth-llcs-1980,shiloach-fcocs-1981,duval-fwoaoa-1983}). We define $\overline{c_i}$ as the lexicographically smallest words among all cyclic permutations of $can(q(c_i))$ and of $can(q(c^{-1}_i))$. We define $\overline{c'_j}$ similarly. The computation of all $\overline{c_i}$ and $\overline{c'_j}$ is done in $O(\ell)$ time.

Now, to test whether the systems of curves $\mathcal C$ and $\mathcal C'$ are freely homotopic up to orientation it remains to test whether the multisets\footnote{Note that several curves in $\mathcal C$ may be homotopic, and similarly for $\mathcal C'$.} $\{\overline{c_1}, \ldots, \overline{c_k}\}$ and $\{\overline{c'_1}, \ldots, \overline{c'_k}\}$ are equal. The naive approach that consists in sorting the elements of each multiset in lexicographic order and test for equality would require $O(\ell^2)$ time. To achieve $O(\ell\log g)$ time, we use the following variant. We consider the prefix tree $T_{\mathcal C}$ corresponding to the set of words $\{\overline{c_1}, \ldots, \overline{c_k}\}$. 
The root of this tree corresponds to the empty word and each other node in the tree represents a prefix of some word in $\{\overline{c_1}, \ldots, \overline{c_k}\}$. We label every edge with a letter in $\{1, \ldots, 8g\}$, so that reading the letters along a root-to-leaf path gives one of the words. For each node of $T_{\mathcal C}$ we store its at most $8g$ children in a binary search tree to guarantee $O(\log g)$ time per search or insertion. We build $T_{\mathcal C}$ inductively from a root by inserting each word $\overline{c_i}$ at a time. If the insertion of $\overline{c_i}$ leads to a leaf already in the tree (meaning that the same word already occurred) we increment a counter for that leaf. The construction of $T_{\mathcal C}$ clearly takes $O(\ell\log g)$ time in total. We similarly build the prefix tree $T_{\mathcal C'}$ for $\mathcal C'$ in the same amount of time.
%a deterministic automaton $A$ that recognizes the finite set $\{\overline{c_1}, \ldots, \overline{c_k}\}$ built as follows. The automaton $A$ is a tree whose set of states are the prefixes of the words in $\{\overline{c_1}, \ldots, \overline{c_k}\}$. There is a transition with letter $a \in \{1, \ldots, 8g\}$ from a word $p$ to the word $p a$ if $p a$ is a prefix of some $\overline{c_i}$. This automaton has size $O(\ell)$ and can be constructed inductively in $O(\ell\log g)$ time. At each state of the automaton, the transitions form a subset of $\{1, \ldots, 8g\}$ and one can use red-black trees to store them and guarantee $O(\log g)$ time per search or insertion. Given the automaton $A_i$ for $\{\overline{c_1}, \ldots, \overline{c_i}\}$, we build the one for $A_{i+1}$. We start by finding the longest prefix of $c_{i+1}$ that is a state in $A_i$. If $c_{i+1}$ is already a state in $A_i$ then we register it as a final state. Otherwise, we add a branch to $A_i$ starting from the state of the prefix and reading the remaining suffix of $c_{i+1}$. The construction of $A$ only requires to read each word $c_i$ once and perform a single search or insertion for each letter. It thus takes $O(\ell\log g)$ time.
%
%We now decorate each terminal state of $A$ with the multiplicity of the corresponding curve in the multiset $\mathcal C$. This step is easily done in $O(\ell\log g)$ time as well.
%The analogous automaton $A'$ with multiplicities is computed for the multiset $\mathcal C'$ in the same amount of time.
Now, the prefix trees (with counters) are equal if and only if the corresponding multisets are equal. The comparison of the prefix trees is easily performed in $O(\ell\log g)$ time by traversing them in parallel. Adding the precomputation cost to the for the construction
and comparison of the prefix trees concludes the proof of the theorem.
\end{proof}

\begin{proof}[Proof of Theorem~\ref{thm:H-H'-equivalence}]
Let $K$ and $K'$ be minor kernels of respectively $H$ and $H'$. By Theorem~\ref{thm:minor-kernel} they can be computed in $O(n^3 \log n)$ time. Let $\mathcal C = (c_1, \ldots, c_k)$ and $\mathcal C' = (c'_1, \ldots, c'_{k'})$ be the systems of closed curves corresponding to the medial of $K$ and $K'$ respectively. By construction $\mathcal C$ and $\mathcal C'$ are tight systems of primitive curves and are obtained by sequences of smoothings from the system of closed curves corresponding to the medial of $G$. By Schrijver's correspondence (see Section~\ref{sec:schrijver-s-correspondence}), we have that $\nu_{\mathcal{C}} =\nu_{\mathcal{C}'}$ if and only if $\mu_H=\mu_{H'}$. In other words, after a reduction taking $O(n^3 \log n)$ time, we are left with the problem of deciding whether the tight system of primitive curves $\mathcal C$ and $\mathcal C'$ have the same $\nu$-spectrum. According to Proposition~7 in~\cite{Sch92} this will be the case if and only if $\mathcal C$ and $\mathcal C'$ are freely homotopic up to orientation.

If $k \neq k'$, then $\mathcal C$ and $\mathcal C'$ are not freely homotopic up to orientation. We therefore assume from now on that $k=k'$.
We rely on Theorem~\ref{thm:homotopy-system-of-curves} for the homotopy comparison. To do so, we represent the systems of curves $\mathcal C$ and $\mathcal C'$ as walks on $M(G)$. Let $m$ and $m'$ be number of smoothings of $G$ to obtain respectively $M(K)$ and $M(K')$ from $M(G)$. Then each curve $c_i$ of length $\ell_i$ in $\mathcal C$ and each curve $c'_j$ of length $\ell'_j$ in $\mathcal C'$ is freely homotopic to a walk on $M(G)$ that can be computed in $O(m \ell_i)$ and $O(m' \ell'_j)$ time respectively. Indeed, each edge in $M(K)$ or $M(K')$ corresponds to a walk in $M(G)$ of length respectively $O(m)$ and $O(m')$. The lift of any closed walk in $M(K)$ or $M(K')$ to $M(G)$ is simply the concatenation of these paths.

Since all of $m$, $m'$, $\ell$ and $\ell'$ are smaller than or equal to $n$, representing $\mathcal C$ and $\mathcal C'$ as closed walks in $M(G)$ can be performed in $O(n^2)$ time. We are now in the setup of Theorem~\ref{thm:homotopy-system-of-curves} that shows that the homotopy equivalence up to orientation takes an additional $O(n^2 \log g)$ time. This cost is negligible compared to the $O(n^3 \log n)$ time of the minor kernel computation and concludes the proof of Theorem~\ref{thm:H-H'-equivalence}. 
\end{proof}

\subsection{Computing the Length Spectrum}\label{sec:spectrum-computation}

Let us recall Theorem~\ref{thm:length-spectrum-computation} from the introduction.

\medskip\noindent
{\sc T{\footnotesize HEOREM}~\ref{thm:length-spectrum-computation}.}
{\it Let $G$ be a graph with $n$ edges embedded on a closed oriented surface $S$ of genus $g \ge 2$. After a preprocessing in $O(n^3 \log n)$ time, for every closed walk $c$ of length $\ell$ on the dual of $G$, $\mu_G([c])$ can be computed in $O(g (n+\ell) \log(n+\ell))$ time.
}

\medskip
The question of computing the length spectrum $\mu_{G}$ for a graph $G$ embedded on a closed surface $S$ has been extensively studied by Colin de Verdière and Erickson in~\cite{CdVE10}. They actually show that $\nu_G([c])$ can be computed in $O(g n \ell \log(n \ell))$ time after $O(g n \log g)$ preprocessing time. However, as explained in Section~\ref{ssec:relations-between-spectra}, the $\mu$- and $\nu$-spectra are deduced from one another by changing the underlying graph appropriately. The computations of $\mu_G([c])$ and of $\nu_{G}([c])$ thus reduce to one another at an additional cost of $O(n + \ell)$ that accounts for building the associated graph and rewriting the curve $c$. This additional cost is negligible compared to the $O(g (n+\ell) \log(n+\ell))$ in our result.

Our complexity improvement in Theorem~\ref{thm:length-spectrum-computation} relies on the minor kernel computation, the Schrijver correspondence, and an efficient computation of crossing numbers.
\begin{proof}[Proof of Theorem~\ref{thm:length-spectrum-computation}]
Given the embedded graph $G$, we can compute a minor kernel $K$ of $G$ in $O(n^3 \log n)$ time thanks to Theorem~\ref{thm:minor-kernel}. The medial $M(K)$ of $K$ corresponds to a tight system of primitive curves $\mathcal{C}$. For a walk $c$ in $G^*$ we have that $2 \mu_G([c]) = 2 \mu_K([c]) = \nu_{\mathcal{C}}([c])$. Moreover, $\nu_{\mathcal{C}}([c])$ is the total geometric intersection number between $[c]$ and the curves in $\mathcal{C}$.

We now explain how to transform a closed walk $c$ in $M(G)^*$ into a (freely homotopic) closed walk in $M(K)^*$ in time $O(|c|_G)$. Recall that $M(K)$ is obtained from $M(G)$ by a sequence of smoothings of corners of empty monogons or minimal bigons. Viewed in the dual graph $M(G)^*$, each smoothing amounts to merge some pairs or triples of edges of $M(G)^*$, while leaving the other edges untouched. We thus get a surjective map from the edges of $M(G)^*$ to the edges of the graph after smoothing. Moreover, this map preserves the free homotopy class of curves. The sequence of smoothings leads by composition to a (surjective) map from the edges of $M(G)^*$ to the edges of $M(K)^*$ that preserves free homotopy classes. We can finally apply this map to the edges of $c$ to obtain a freely homotopy closed walk of the same length in $M(K)^*$.

We are now left with the computation of the geometric intersection number between a curve $c$ of length $O(\ell)$ written as a closed walk on $M(K)^*$ and the system of closed curves $\mathcal{C}$. Building on~\cite{DL19}, Dubois~\cite{Dub24} designed an efficient algorithm to solve this question that runs in $O(g (n + \ell) \log (n + \ell))$ time. This concludes the proof.
\end{proof}

\appendix
\section{Area of a Bigon}\label{sec:area-of-bigon}

Let $\mathcal C$ be a system of closed curves in a closed oriented surface $S$ of genus $g \ge 2$. The goal of this appendix is to provide an upper bound on the area of bigons of $\widetilde{\mathcal C}$ in terms of their lengths. Note that contrarily to most of the sections in this article, we do not require for the bigons to be minimal.

\begin{theorem}
\label{area length}
Let $\mathcal C$ be a filling system of closed curves with $n$ vertices on a surface $S$ of genus $g\ge 2$. If $b$ is a bigon of $\widetilde{\mathcal C}$ of length $\ell$ and area $A$, then \[A\le\frac{4g}{4g-6}(n+2-2g)\ell\le\frac{4g}{4g-6}n\ell.\]
\end{theorem}

To prove Theorem~\ref{area length}, we rely on the following isoperimetric inequality. Recall that the embedded graph $G$ is reduced if it has a single vertex and a single face.
\begin{lemma}
\label{Rbound}
Let $G$ be a reduced graph cellularly embedded in a closed surface $S$ of genus $g\ge 2$. Let $R$ be a finite set of faces in $\widetilde G$. Then
 the number $|\partial R|$ of boundary edges of $R$ and the number $|R|$ of faces in $R$ satisfies the following inequality
 \[
 |R| \le \frac{1}{4g-6} |\partial R|.
% |\partial R|\ge (4g-6).
 \]
\end{lemma}
\begin{proof}
We prove the inequality by induction on $|R|$.
If $|R|=1$, then $R$ contains a single domain and $|\partial R|=4g$ implying the inequality in the lemma.
For the inductive step, assume that $|\partial R'|\ge (4g-6)|R'|$ for all regions $R'$ such that $|R'|\le n$ and consider a region $R$ such that $|R|=n+1$.
Following Dehn~\cite{Sti87}, we fix a vertex $v_0\in R$ and denote by $R_1$ the set of domains incident to $v_0$. We then define annular regions $R_{i}$ inductively, so that $R_i$ is the set of domains not in $R_{i-1}$ but sharing a vertex with $R_{i-1}$ (see Figure~\ref{DehnAlgo}). For every domain $D\in R_i$, $D$ has either one or two vertices in $R_i\cap R_{i+1}$ and has two edges shared with domains in $R_i$. It follows that $D$ has at least $4g-3$ edges on $R_i\cap R_{i+1}$.

\begin{figure}[ht!]
    \centering
	\def\svgwidth{5cm}
	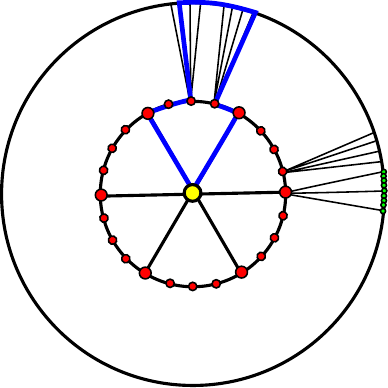
    \caption{The faces of the tiling are arranged in concentric rings}
    \label{DehnAlgo}
\end{figure}

Let $i$ be the maximal index such that $R_i\cap R\neq\varnothing$ and consider a face $C$ in this intersection. Then, by the induction hypothesis $|\partial(R\backslash C)|\ge (4g-6)n$. By the maximality of $i$ the face $C$ has at least $4g-3$ sides on the external boundary of $R_i$ and $C$ has at most $3$ sides in common with $R\backslash C$. So $|\partial R|\ge |\partial(R\backslash C)|+4g-3-3\ge(4g-6)(n+1)$.
\end{proof}

\begin{proof}[Proof of Theorem~\ref{area length}]
Let $(T,L,C)$ be a tree-cotree decomposition (see~\cite{Epp02}) where $T$ is a spanning tree of $\mathcal C$, $C^*$ is a spanning tree of $\mathcal{C}^*$ the dual of $\mathcal C$ with $T\cap C=\varnothing$ and $L=E\backslash(C\cup T)$ is the set of leftover edges. By Euler's formula, $|L|=2g$. Moreover, $S$ cut through the edges of $T\cup L$ is a topological disk.

Let $G$ be the graph obtained from $\mathcal C$ by contracting all the edges of $T$. Note that this operation preserves the number of faces. Let $b'$ be the image of $b$ through this contraction. Then $b'$ has the same area as $b$ and a smaller length $\ell'\le \ell$. Let $G_L$ be the graph induced from $G$ by $L$. By construction, $G_L$ is a graph with one single vertex and one single face. Let $R$ be the set of all the faces of $G_L$ intersecting $b'$. Moreover, each face of $R$ with an edge of $\partial R$ contains an edge of the boundary of $b'$. So the number of boundary face of $R$ is at most $\ell'$. As each face has $4g$ edges, we get $|\partial R|\le 4g\ell'$.

By Lemma~\ref{Rbound}, $(4g-6)|R|\le |\partial R|\le 4g\ell'$. So $|R|\le \frac{4g}{4g-6}\ell'$. Moreover, as $\mathcal C$ is in general position, each vertex has degree $4$. So by Euler's formula the number of faces of $G$ is $n+2-2g$. As $b'\subset R$ and each face of $\widetilde {G_L}$ contains $n+2-2g$ faces of $\widetilde G$, $A\le (n+2-2g)\cdot|R|\le \frac{4g(n+2-2g)}{4g-6}\ell'$. Recall that $\ell'\le \ell$, we deduce that $A\le \frac{4g(n+2-2g)}{4g-6}\ell\le \frac{4g}{4g-6}n\ell$.
\end{proof}

\section{Isoperimetric Inequality for System of Quads}
In this section we prove an edge isoperimetric inequality for system of quads that is used to bound the complexity of the simplicity test of Section~\ref{sec:simplicity-test}. It is of the same flavor as Theorem~\ref{Rbound} in the previous appendix.

Recall that a system of quads is a quadrangulation of a closed surface $S$ that is obtained from a reduced graph on $S$ by adding a new vertex in the center of its unique face, adding $4g$ edges from this central vertex in order to make a triangulation and removing the $2g$ edges of the initial reduced graph.

\begin{theorem}
\label{th:isoperimetric-quad-system}
Let $Q$ be a system of quads of a surface of genus at least $2$. Let $H$ be a finite subgraph of $\widetilde{Q}$. Then $|E(H)|\le \frac{5}{2} |\partial H|$ where $\partial H$ is the set of boundary edge sides of $H$.
\end{theorem}
In the following, we assume that $H$ is a non-empty finite subgraph of $\widetilde Q$. We fix a root vertex $v_0$ in $H$ and orient the edges of $\widetilde{Q}$ towards $v_0$. We let $F(H)$ denote the set of faces of $\widetilde{Q}$ such that its 4 adjacent edges belong to $H$ and call its elements the \emph{faces of $H$}. Each side of an edge of $H$ belongs to a face of $H$ or is called a \emph{boundary edge side}. The set of boundary edge sides of $H$ is denoted $\partial H$. We say that a face $f$ of $\widetilde{Q}$ has \emph{level $i$} if its vertices are at distances $(i-1,i,i+1,i)$ from the root $v_0$.

\begin{claim}
\label{cl:quad-face-oriented-degrees}
Each face of $\widetilde{Q}$ has a well defined level $i$ and is adjacent to exactly two faces of level $i+1$ and at most one face of level $i-1$.
\end{claim}

\begin{proof}
Because $\widetilde{Q}$ is bipartite, the distances of the vertices of $f$ to the root is either of the form $(i-1,i,i+1,i)$ or $(i-1,i,i-1,i)$ for some $i \ge 1$. The fact that the second form is impossible follows from~\cite{LR12}.
Let now $f$ be a face of level $i$. Using that every vertex $v$ at distance $i$ is adjacent to at most two vertices at distance $i-1$, and that together with $v$ they share a same face, we infer that:
\begin{itemize}
    \item The two edges of $f$ with endpoints at distances $i$ and $i+1$ are adjacent to faces of level $i+1$.
    \item For the two edges of $f$  with endpoints at distances $i$ and $i-1$ we have two possibilities: either both of them are adjacent to faces at level $i$ or one is adjacent to level $i-1$ and the other to level $i$.
\end{itemize}
\end{proof}

\begin{claim}
\label{cl:face-isoperimetric}
$|F(H)| \le |\partial H|$.
\end{claim}

\begin{proof}
Our goal is to build a function $\phi: F(H) \to \mathbb{R}^{\partial H}$ such that
\[
\forall f \in F(H), \sum_{e \in \partial H} \phi(f)_e = 1
\quad \text{and} \quad
\forall e \in \partial H, \sum_{f \in F(H)} \phi(f)_e \le 1.
\]
One has to think of $\phi$ as a way of discharging the faces on the boundary edges, namely $\phi(f)_{e}$ should be thought of as the weight $f$ gives to $e$. If such function $\phi$ exists, it implies the desired inequality as by the first item $|F(H)| = \sum_{f \in F(H)} \sum_{e \in \partial H} \phi(f)_e$ and by the second item $|\partial H| \ge \sum_{e \in \partial H} \sum_{f \in F(H)} \phi(f)_e$.

In order to define $\phi$, we use an auxiliary oriented graph on $F(H) \sqcup \partial H$ built as follows. We put an oriented edge from $f \in F(H)$ to $f' \in F(H)$ if $f$ and $f'$ are adjacent and the level of $f'$ is one higher than the one of $f$. We also put an oriented edge from $f \in F(H)$ to $e \in \partial H$ if $e$ is adjacent to $f$ and the face across $e$ has a level one higher than the one of $f$. By Claim~\ref{cl:quad-face-oriented-degrees}, any $f \in F(H)$ has out-degree 2 and in-degree at most 1. Now for $f \in F(H)$ and $e \in \partial H$, we define
\[
\phi(f)_e :=
\left\{ \begin{array}{ll}
2^{-k} & \text{if there is a directed path of length $k$ from $f$ to $e$,} \\
0 & \text{otherwise.}
\end{array} \right.
\]
We claim that $\phi$ satisfies the two requirements. Indeed, because the out-degree is 2, for any $f \in F(H)$
\[
\sum_{e \in \partial H} \phi(f)_e = 1.
\]
Next, because the in-degree is at most one, we have that for any $e \in F(H)$
\[
\sum_{f \in F(H)} \phi(f)_e \le \sum_{n \ge 1} \frac{1}{2^n} = 1.
\]
\end{proof}

We now prove Theorem~\ref{th:isoperimetric-quad-system}. 
We first prove that $4|F(H)|+|\partial H|=2|E(H)|$. This is indeed a simple double counting argument. The quantity $2 |E(H)|$ is the number of edge sides in $H$. Now each edge side is either a boundary or inside a face. Since faces are quadrilaterals, we get 4 edge sides per faces.

Now, applying the inequality $F(H)\le|\partial H|$ from Claim~\ref{cl:face-isoperimetric} we get that $E(H)\le \frac{5}{2}|\partial H|$ proving the theorem.

\section*{Acknowledgments.}
 We would like to thank the referees for their thorough reviews whose numerous comments allowed to improve the presentation of the preliminary version. We also thank Arnaud de Mesmay for pointing to us that~\cite{CdVE10}
 applies to more general graphs than we do.

Vincent Delecroix and Oscar Fontaine are partially supported by ANR MOST (ANR-23-CE40-0020) and ANR CarteEtPlus (ANR-23-CE48-0018).

\bibliography{main.bib}
%\printbibliography

\end{document}

%% file: DehnAlgo.pdf_tex
%% Creator: Inkscape 1.2.2 (b0a8486541, 2022-12-01), www.inkscape.org
%% Accompanies image file 'DehnAlgo.pdf' (pdf, eps, ps)
\begingroup%
  \makeatletter%
  \providecommand\color[2][]{%
    \errmessage{(Inkscape) Color is used for the text in Inkscape, but the package 'color.sty' is not loaded}%
    \renewcommand\color[2][]{}%
  }%
  \providecommand\transparent[1]{%
    \errmessage{(Inkscape) Transparency is used (non-zero) for the text in Inkscape, but the package 'transparent.sty' is not loaded}%
    \renewcommand\transparent[1]{}%
  }%
  \providecommand\rotatebox[2]{#2}%
  \newcommand*\fsize{\dimexpr\f@size pt\relax}%
  \newcommand*\lineheight[1]{\fontsize{\fsize}{#1\fsize}\selectfont}%
  \ifx\svgwidth\undefined%
    \setlength{\unitlength}{186.08860036bp}%
    \ifx\svgscale\undefined%
      \relax%
    \else%
      \setlength{\unitlength}{\unitlength * \real{\svgscale}}%
    \fi%
  \else%
    \setlength{\unitlength}{\svgwidth}%
  \fi%
  \global\let\svgwidth\undefined%
  \global\let\svgscale\undefined%
  \makeatother%
  \begin{picture}(1,0.99832282)%
    \lineheight{1}%
    \setlength\tabcolsep{0pt}%
    \put(0,0){\includegraphics[width=\unitlength,page=1]{DehnAlgo.pdf}}%
    \put(0.46239893,0.40205499){\color[rgb]{0,0,0}\makebox(0,0)[lt]{\lineheight{0}\smash{\begin{tabular}[t]{l}$v_0$\end{tabular}}}}%
    \put(0.10626551,0.59277118){\color[rgb]{0,0,0}\makebox(0,0)[lt]{\lineheight{0}\smash{\begin{tabular}[t]{l}$R_2$\end{tabular}}}}%
    \put(0.62392364,0.81657075){\color[rgb]{0,0,1}\makebox(0,0)[lt]{\lineheight{0}\smash{\begin{tabular}[t]{l}$\tilde{c}$\end{tabular}}}}%
  \end{picture}%
\endgroup%